\documentclass[10pt, final, conference, letterpaper, oneside, twocolumn]{resources/IEEEtran}
\usepackage[T1]{fontenc}
\usepackage[utf8]{inputenc}
\usepackage[english]{babel}
\AtBeginDocument{%
\providecommand\BibTeX{{%
			\normalfont B\kern-0.5em{\scshape i\kern-0.25em b}\kern-0.8em\TeX}}}

\usepackage{hyperref}
\usepackage{cite}

\usepackage{amsmath,amssymb,amsfonts} 
\usepackage[noend]{algpseudocode}
\usepackage{algorithmicx, algorithm}
\usepackage{graphicx} 
\usepackage{textcomp}
\usepackage{xcolor}

\usepackage{tikz}
\usetikzlibrary{arrows,decorations.markings, positioning, patterns, fit}
\usepackage{subcaption}
\usepackage{tabu}

\usepackage{resources/mathcommands}

\usepackage{resources/scheduling}

\usepackage{amsthm}

\usepackage[switch]{lineno}

\newtheorem{thm}{Theorem}

\newtheorem{cor}[thm]{Corollary}

\theoremstyle{definition}
\newtheorem{defn}[thm]{Definition}

\theoremstyle{remark}

\newcommand{\parlabel}[1]{\smallskip\noindent\textbf{#1}\hspace{0.3em}}

\begin{document}

\title{Parallel Path Progression DAG Scheduling}

\author{
\IEEEauthorblockN{Niklas Ueter}
\IEEEauthorblockA{\textit{TU Dortmund University}} 
Dortmund, Germany \\
niklas.ueter@tu-dortmund.de
\and
\IEEEauthorblockN{Mario G\"unzel}
\IEEEauthorblockA{\textit{TU Dortmund University}} 
Dortmund, Germany \\
mario.guenzel@tu-dortmund.de
\and
\IEEEauthorblockN{Georg von der Br\"uggen}
\IEEEauthorblockA{\textit{TU Dortmund University}} 
Dortmund, Germany \\
vdb@tu-dortmund.de
\and
\IEEEauthorblockN{Jian-Jia Chen}
\IEEEauthorblockA{\textit{TU Dortmund University}} 
 Dortmund, Germany \\
jian-jia.chen@tu-dortmund.de
}

\maketitle

\begin{abstract} 
 To satisfy the increasing performance needs of modern cyber-physical systems, multiprocessor architectures are increasingly utilized.  
 To efficiently exploit their potential parallelism in hard real-time systems,  
 
 appropriate task models and scheduling  
 algorithms that allow providing timing guarantees are required. 
 Such scheduling algorithms and the corresponding worst-case response time analyses usually suffer from 
 resource over-provisioning due to pessimistic analyses based on worst-case assumptions.  
 Hence, scheduling algorithms 
 and analysis with high resource efficiency are required. 
  A prominent parallel task model is the 
 directed-acyclic-graph (DAG) task model, where precedence constrained subjobs 
 express parallelism.
 
 This paper studies the real-time scheduling problem of sporadic arbitrary-deadline 
 DAG tasks. 
 We propose a path parallel progression scheduling property with only two distinct subtask priorities, which  allows to track the execution of a 
 collection of paths simultaneously. This novel approach significantly improves 
 the state-of-the-art response time analyses for parallel DAG tasks for highly parallel 
 DAG structures. 
 Two hierarchical scheduling 
 algorithms are designed based on this property, extending the parallel path progression properties and improving the response time
 analysis for sporadic arbitrary-deadline DAG task sets.
\end{abstract}

\section{Introduction}
\label{sec:introduction}

Modern cyber-physical systems have gradually shifted from uniprocessor to 
multiprocessor systems in order to deal with thermal and energy constraints, 
as well as the computational demands of increasingly complex applications, like control, perception algorithms, 
 or video processing algorithms. 
This shift poses  
multiple challenges  
for the real-time verification, e.g., 
how to efficiently utilize the parallelism provided by multiprocessors for task sets 
with inter- and intra-task parallelism  
while guaranteeing that each task meets its deadline.
While inter-task parallelism refers to the parallel execution of distinct tasks, 
each of which executes sequentially, intra-task parallelism refers to 
the parallel execution of a single task. Intra-task parallelism requires task models 
with subtask level granularity that can be scheduled in parallel, e.g., 
Fork-join models~\cite{Lakshmanan:2010:SPR:1935940.1936239}, 
synchronous parallel task models, or DAG (directed-acyclic graph) based task models.

A plethora of real-time scheduling algorithms and response time analyses thereof 
have been proposed, e.g., for generalized parallel task models~\cite{SaifullahRTSS2011}, and for DAG (directed-acyclic graph) based task models~\cite{He2019, DongRTSS,DBLP:conf/rtss/00040BB020, Fonseca:2017:RTNS, DBLP:conf/sies/FonsecaNNP16,  DBLP:conf/ipps/Baruah15, DBLP:conf/ecrts/BonifaciMSW13, MelaniECRTS2015}.
For DAG-based task models, improvements in the response time analyses can be categorized into 
analyses that improve inter-task interference, e.g., in~\cite{DongRTSS, Fonseca:2017:RTNS}, or intra-task interference 
as e.g., in~\cite{Li:ECRTS14, He2019, DBLP:conf/ecrts/HeLG21, DBLP:conf/rtss/00040BB020}.
In general, intra-task interference analyses build upon the interference analysis 
along the execution of the envelope (or critical path).
The envelope (or critical path) of a DAG is a schedule-dependent sequence of subjobs of the DAG.
The sequence starts with the last finishing subjob in a concrete schedule and is iteratively composed 
by choosing the subjob that finishes latest among all directly preceding subjobs of the current subjob under consideration. 
When no further predecessors exist, the envelope is complete. 
Since these subjobs \emph{envelope} the execution of the DAG, the interference of the execution of 
this sequence (or path) bounds the response time.
In federated scheduling~\cite{Li:ECRTS14}, the intra-task interference of the envelope 
execution is upper-bounded by the workload of the non-envelope subjobs divided by the number 
of processors. The corresponding response time analysis requires no information about the 
internal structure of the DAG except for the total workload and the longest path.

This analysis was improved by He et al.~\cite{He2019}, who proposed a specific intra-node priority assignment  
for list-scheduling that  
must respect 
the topological ordering of the nodes within the DAG. This priority assignment and the inspection 
of the DAG structure  
results in a less pessimistic upper-bound for a task's self-interference 
of the envelope path compared to federated scheduling. 
These results are further improved and extended by Zhao et al.~\cite{DBLP:conf/rtss/00040BB020}, where
subjob dependencies are explicitly considered along the execution of the envelope path to more accurately bound 
self-interference. 
Most recently, He et al.~\cite{DBLP:conf/ecrts/HeLG21} improve their prior work  
by lifting the topological order restrictions in their intra-node priority assignments, 
which
further improved the results by Zhao et al.~\cite{DBLP:conf/rtss/00040BB020}.

In contrast to the state-of-the-art, we  
consider the progress of \emph{many parallel paths} 
instead of only the envelope in the response time analysis. 
The number of parallel paths 
can amount to the degree of maximum parallelism, i.e., the number of processors.
Since paths, i.e., a sequence of directly preceding subtasks, inherently limit the parallelism of a DAG, 
a path-based analysis instead of a subtask-based analysis is beneficial. 
More precisely, inversely to prior approaches, we do not track the execution progress by the analysis of envelope path interference, 
but track analyzable simultaneous progress of a collection of \emph{many parallel paths} 
alongside the envelope path using intra-task prioritization. By virtue of this novel approach, 
we only have to account for the interference of subjobs that do not belong to any of these parallel paths for 
a response time bound. 

To the best of our knowledge, this is the first paper that proposes \emph{parallel path progression} 
concepts and corresponding response time analyses as well as extensions to 
handle sporadic arbitrary-deadline DAG tasks.
Notably, our work is a generalization of the hierarchical DAG scheduling approach 
by Ueter et al.~\cite{Ueter2018}, which considers  the special case 
when the number of considered parallel paths is reduced to one.

\parlabel{Contributions.}
We provide the following contributions:
\begin{itemize}
\item Based on \emph{Parallel Path Progression Concepts} in Section~\ref{sec:parallel-path-progression-concepts}, 
  we propose a scheduling algorithm and a sustainable (or anomaly free) response time analysis 
  for an arbitrary collection of paths of at most the number of processors in Section~\ref{sec:parallel-path-progression-scheduling}.
 In Section~\ref{sec:n-path-collection-approximation-algorithm}, we provide an approximation algorithm to 
 find a path collection that allows to prove a bounded worst-case response-time with respect to an optimal response-time.

\item We extend our findings to two hierarchical scheduling algorithms in Section~\ref{sec:hierarchical-scheduling}.
  Namely a sporadic arbitrary-deadline gang reservation system in Section~\ref{sec:gang-reservation} and a sporadic 
  arbitrary-deadline ordinary reservation system in Section~\ref{sec:ordinary-reservation} that make use of the \emph{Parallel Path Progression Concepts}. 
  The hierarchical scheduling algorithm can be applied to sporadic arbitrary-deadline DAG tasks, which may be executed concurrently with tasks described by a different task model.

\item For both reservation systems, we provide response time analyses and algorithms to 
  generate feasible reservation systems in Section~\ref{sec:gang-reservation} and Section~\ref{sec:ordinary-reservation}, 
  respectively.

\item In Section~\ref{sec:evaluation}, we evaluate our approach using synthetically generated 
  DAG task sets and demonstrate that our approach outperforms the state-of-the art in high-parallelism 
  scenarios and show that the performance of our approach is between the start-of-the-art and federated scheduling 
  in more sequential scenarios.
\end{itemize}
 
\section{Task Model and Problem Description}
\label{sec:task-model-and-problem-description}

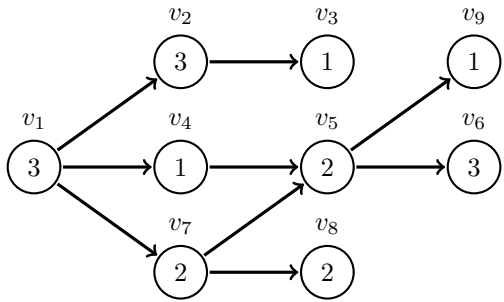
\begin{figure}[b]
\centering
\pgfdeclarelayer{layer1}
\pgfdeclarelayer{layer2}
\pgfdeclarelayer{layer3}
\pgfdeclarelayer{layer4}
\pgfdeclarelayer{foreground}
\pgfsetlayers{layer4,layer3,layer2,layer1,main,foreground}

\begin{tikzpicture}[]
  \def\ux{1.95cm}\def\uy{1.4cm}
        \tikzset{
           job/.style={minimum height=0.5*\uy, circle, draw, inner sep=2pt, outer sep = 1pt, thick}
        }

\node[job,  anchor=south west, label=above:{$v_1$}, draw] at (0, 0)  (V1) {$3$};
\node[job, anchor=south west, label=above:{$v_2$}, draw] at (1*\ux, 1*\uy)  (V2) {$3$};
\node[job,  anchor=south west, label=above:{$v_3$}, draw] at (2*\ux, 1*\uy)  (V3) {$1$};
\node[job, anchor=south west, label=above:{$v_4$}, draw] at (1*\ux, 0)  (V4) {$1$};
\node[job, anchor=south west, label=above:{$v_5$}, draw] at (2*\ux, 0)  (V5) {$2$};
\node[job, anchor=south west, label=above:{$v_6$}, draw] at (3*\ux, 0)  (V6) {$3$};

\node[job, anchor=south west, label=above:{$v_7$}, draw] at (1*\ux, -1*\uy)  (V7) {$2$};
\node[job, anchor=south west, label=above:{$v_8$}, draw] at (2*\ux, -1*\uy)  (V8) {$2$};

\node[job, anchor=south west, label=above:{$v_9$}, draw] at (3*\ux, 1*\uy)  (V9) {$1$};

\path[->, very thick]	(V1)		edge	[]	(V2);
\path[->, very thick]	(V2)		edge	[]	(V3);

\path[->, very thick]	(V1)		edge	[]	node [above] {}	(V4);
\path[->, very thick]	(V4)		edge	[]	node [above] {}	(V5);

\path[->, very thick]	(V5)		edge	[]	node [above] {}	(V6);

\path[->, very thick]	(V1)		edge	[]	node [above] {}	(V7);
\path[->, very thick]	(V7)		edge	[]	node [above] {}	(V8);

\path[->, very thick]	(V7)		edge	[]	node [above] {}	(V5);

\path[->, very thick]	(V5)		edge	[]	node [above] {}	(V9);

\end{tikzpicture}
\caption{An exemplary directed-acyclic graph (DAG) with subtasks $v_1, v_2, \dots, v_9$. 
The  
numbers within the 
nodes denote the subjob's worst-case execution time.
The arrows represent the precedence constraints  
indicating that the 
release of a subjob depends on the finishing of all incident subjobs.}
\label{fig:example-dag-plain}
\end{figure}
 
We consider a finite set $\mathbb{T} := \setof{\tau_1, \dots, \tau_z}$ 
of sporadic arbitrary-deadline directed-acyclic graph (DAG) tasks 
that are scheduled and executed upon $M$ homogeneous processors.
Each task \mbox{$\tau_i := (G_i, D_i, T_i) \in \mathbb{T}$} is defined by a DAG $G_{i}$ describing the subtasks 
and precedence constraints, minimal inter-arrival time $T_{i}$, and relative deadline $D_{i}$.
Each task releases an infinite sequence of task instances, called
jobs. We  
use $J_{i}^{\ell}$ to denote the $\ell$-th job of task $\tau_i$, and  
$a_i^{\ell}$, $f_i^{\ell}$, and $d_i^{\ell}=a_i^{\ell}+D_i$ to refer to the arrival time, finishing time, 
and (absolute) deadline of job $J_{i}^{\ell}$.

\parlabel{DAG.} 
The task's DAG $G_i$ is defined by the tuple $(V_i,E_i)$, where $V_i$ denotes the finite set
of subtasks and the relation $E_i \subseteq V_i \times V_i$ denotes the precedence constraints 
among them, such that there are no cyclic precedence constraints. 
To be mathematically precise, each job $J_i^{\ell}$ is associated with an instance of the DAG 
$G_i^{\ell}$ with corresponding $\ell$-th subjobs $v_j^{\ell}$ where $v_j$ is a subtask in $V_i$.
A subjob of the $\ell$-th job of task $\tau_i$, namely $v_j^{\ell}$ for $v_j \in V_i$, is released when all $\ell$-th 
subjobs $v_k^{\ell}$ for $(v_k, v_j) \in V_i$ have finished execution.
To reduce this notation, we drop the index of the task as well as of the job when analyzing one specific job.
That is, we refer to $G=(V,E)$ and $v_j \in V$ to denote a subjob of a specific DAG job.
An exemplary DAG is illustrated in Figure~\ref{fig:example-dag-plain}.

\parlabel{Volume.}
The volume \mbox{$vol_i: V_i \to \mathbb{N}_{\geq 0}$} specifies the worst-case execution time of each subtask $v_j \in V_i$, 
which means that no subjob (instance) $v_j^{\ell}$ ever executes for more than $vol_i(v_j)$ 
time-units on 
the execution platform, but may finish earlier. 
Moreover, the  
volume of any subset of subtasks $W \subseteq V_i$ 
is  
$vol(W) := \sum_{v_j \in W} vol_i(v_j)$.
In particular, the \emph{total volume} of a task is given by $C_i := vol_i(V_i)$.

\parlabel{Release \& Deadline.}
In real-time systems,
tasks must fulfill timing requirements, i.e., each job must finish
at most its total volume between the arrival of a job
at $a_{i}^{\ell}$ and that job's absolute deadline at $a_{i}^{\ell} + D_i$.
A task $\tau_i$ is said to meet its deadline if each job meets its deadline, i.e., 
$f_{i}^{\ell} \leq a_{i}^{\ell}+D_i$ for all $\ell \in \mathbb{N}$. 
We consider \emph{arbitrary-deadlines}, which means that 
we do not make any assumptions about the relation of deadline and inter-arrival time. 
For example, the relative deadline may be less than the minimal inter-arrival time $(D_i \leq T_i)$ 
in which case a new job is only released if the previous job is finished. 
Alternatively, the deadline can be larger than the minimal inter-arrival time $(D_i > T_i)$ in which 
case a new job can be released despite an unfinished prior job.
The release times of any two subsequent jobs of $\tau_i$ is at least $T_i$ time-units 
apart, i.e., \mbox{$a_{i}^{\ell+1} \geq a_{i}^{\ell}+T_i$} for all \mbox{$\ell \in \mathbb{N}$}.

\section{Parallelism \& Path-Aware Scheduling}
\label{sec:parallelism-and-path-aware-scheduling}

The parallelism of a DAG task is inherently limited by the (simple) paths it is composed of, 
since a path enforces a sequential execution order of the subjobs along the path.
In particular, the intertwining of paths degrades parallelism; for example, two disjoint paths 
allow for fully parallel execution, while two intertwined paths result in limited parallel execution.
In this paper, we aim to reduce intra-task interference by enforcing properties  
that track and guarantee parallel progress of a collection of paths within a DAG and thus 
allow to significantly improve the worst-case response time for high-parallel use cases.

To that end, we answer the following questions:
(\emph{1})~What are the minimal theoretical properties required to track and guarantee parallel path progression on a set of dedicated processors?
(\emph{2})~Can any collection of paths be used for parallel progression and if so what is an optimal selection?
(\emph{3})~How can the results from (\emph{1}) and (\emph{2}) be extended to consider inter-task interference? 

\subsection{Parallel Path Progression Concepts}
\label{sec:parallel-path-progression-concepts}

In this subsection, we examine the required properties to achieve parallel path progression on $M$ processors 
dedicated to execute a single job of a DAG task. By that, we avoid any inter-task interference 
and solely focus on intra-task interference.

\begin{defn}[Path]
\label{def:path}
 Let $G=(V,E)$ denote a DAG then for each subtask $v_j \in V$ the set of predecessors of $v_j$
 and the set of successor of $v_j$ is given by $pred(v_{j}) := \setof{v_{i} \in V~|~(v_{i}, v_{j}) \in E}$ and
  $succ(v_{j}) := \setof{v_{i} \in V~|~(v_{j}, v_{i}) \in E}$, respectively.
 A path in $G$ is a partially ordered set of subtasks $\pi := \setof{v_{1}, \dots, v_{k}, \dots, v_{n}}$
 such that $pred(v_{1}) = \emptyset$, $succ(v_{n}) = \emptyset$ 
 and $v_{k} \in pred(v_{k+1})$ for all $k \in \setof{1, \dots, n-1}$.
 If either $pred(v_{1}) \neq \emptyset$ or $succ(v_{n}) \neq \emptyset$ 
 holds then $\pi$ is considered a sub-path.
\end{defn}

\begin{defn}[$n$-Path Collection]
\label{def:n-path-collection} 
Given a DAG $G = (V,E)$ the enumeration of all possible paths in $G$ namely 
$\Psi(G) := \setof{\pi~|~\pi~\text{is a path according to Definition~\ref{def:path} in}~G}$ can be computed in exponential time.
Any subset of paths  \mbox{$\psi \in \mathcal{P}(\Psi(G))$} from the powerset of $\Psi(G)$ is called a \emph{path collection}.
 A path collection $\psi \in \mathcal{P}(\Psi(G))$ is an \emph{$n$-path collection} 
 if $|\psi| = n$, i.e., a collection of $n$ paths.
\end{defn}

It is intuitively clear that the maximal number of paths 
that can be executed in parallel is limited by the number of parallel processing elements, i.e., the number of processors $M$.
Therefore we constrain our solution space to $n$-path collections where 
$n \in \setof{1, \dots, M}$.
Based on a concrete $n$-path collection $\psi$, the set of subtasks that belong to at least one of the paths 
in $\psi$ is defined by $V_s(\psi) := \pi_{\psi_1} \cup \dots \cup \pi_{\psi_n}$ for each $\pi_{\psi_1}, \dots, \pi_{\psi_n} \in \psi$.
Conversely, the complement set of subtasks that do not belong to any of the selected paths 
is denoted by $V_s^c(\psi) := \setof{v \in V~|~v \notin V_s(\psi)}$.
We propose a parallel-progress prioritization that gives each subtask a priority based on the membership of the above sets, 
which is formalized in the following definition. 
Later in this section, we explain how this prioritization can be used to 
better analyze the self-interference by explicitly considering the 
parallel execution of paths in $\psi$ in the response-time analysis.

\begin{defn}[Parallel Path Progression Prioritization]
\label{def:parallel-progress-prioritization}
 Let $V_s(\psi)$ denote the set of 
 subtasks from an $n$-path collection $\psi$ of a DAG $G=(V,E)$.
 A fixed-priority policy for all subtasks $v \in V$ is a \emph{parallel path progression prioritization} 
 if and only if $\Pi(v_i) < \Pi(v_k)~\text{for any two}~ v_i \in V_s(\psi)~\text{and}~v_k \in V_s^c(\psi)$, where $\Pi(v_i)$ 
 denotes the priority of subtask $v_i$.
\end{defn}

Note that in our notation for the priorities,   
a higher value of $\Pi$ implies a higher priority, i.e., $\Pi(v_i) > \Pi(v_k)$ implies 
that $v_i$ has a higher priority than $v_k$.
A sufficient policy to satisfy the parallel path progression prioritization property 
is to only use two distinct priority-levels. 

For clarity of notation, we elaborate the above definitions collectively 
in the following example. 
The DAG in Figure~\ref{fig:example-dag-plain} can be exhaustively decomposed into six paths, i.e., 
$\Psi(G) := \setof{\pi_1, \pi_2, \dots, \pi_6}$.
The individual paths are: 
$\pi_1 := \{v_1, v_2, v_3\}$, $\pi_2 := \{v_1, v_4, v_5, v_9\}$, $\pi_3 := \{v_1, v_4, v_5, v_6\}$, 
$\pi_4 := \{v_1, v_7, v_5, v_9\}$, $\pi_5:=\{v_1, v_7, v_5, v_6\}$, and $\pi_6 := \{v_1, v_7, v_8\}$.
A $2$-path collection $\psi$ from the powerset $\mathcal{P}(\Psi(G))$ is for instance given by 
$\psi := \setof{\pi_2, \pi_3}$. Subsequently, $V_s(\psi) = \pi_2 \cup \pi_3 = \setof{v_1, v_4, v_5, v_6, v_9}$ 
and $V_s^c(\psi) := \setof{v_2, v_3, v_7, v_8}$.  
If for instance all subjobs
$v_i \in V_s(\psi)$ 
are assigned priority $\Pi(v_i) = 1$ and conversely all subjobs $v_i \in V_s^c(\psi)$ 
are assigned priority $\Pi(v_i) = 2$, then this prioritization is a valid parallel path progression prioritization.

\subsection{Parallel Path Progression Scheduling}
\label{sec:parallel-path-progression-scheduling}

In
this subsection, we look at a single DAG job that  
is scheduled on $M$ dedicated processors by a work-conserving preemptive and non-preemptive 
list scheduling 
algorithm in conjunction with the parallel path progression prioritization. We 
elaborate how this prioritization aids the analysis of the 
parallel progression of a path collection and in consequence the response time analysis of 
the DAG job. Furthermore, we show the sustainability of our timing analysis with respect to timing anomalies.  

\begin{defn}[List-FP]
\label{def:list-fp-dispatching}
 In a \emph{preemptive list-FP schedule} on $M$ dedicated processors, 
 a task instance (job) of a DAG task \mbox{$G=(V,E)$} with a fixed-priority assignment 
 of each each subjob $v \in V$ is scheduled according to the 
 following rules:
 \begin{itemize}
   \item A subjob arrives to the ready list if all preceding subjobs have executed until completion, i.e., 
     the subjob arrival time $a_{i}$ for each subjob $v_{i}$ is given by $\max\setof{f_{j}~|~v_{j} \in pred(v_{i})}$. 
     An arrived but not yet finished subjob is considered pending.
   \item At any time $t$, the $M$ highest-priority pending subjobs are executed on the $M$ processors 
     and a lower-priority subjob is preempted if necessary.
 \end{itemize}
For a \emph{non-preemptive list-FP schedule}, we assume that each subjob is non-preemptible, i.e., 
runs to completion once it started executing.
\end{defn}

Note that the required theoretical properties for parallel path progression 
are satisfied using only two distinct priorities and are thus satisfied regardless of 
which lower-priority subjob is preempted at any point in time.
From a practical standpoint however, unnecessary preemptions of different subjobs should be avoided by, 
e.g., favoring to preempt the subjobs that have already been preempted before to retain 
cache coherence.

Before we are able to provide a response time analysis, we introduce and formalize the concept of an envelope 
of a schedule.

\begin{defn}[Envelope]
\label{def:envelope}
 Let $S$ be any concrete schedule of the subjobs $V=\setof{v_1, \dots, v_\ell}$ of a given DAG job of some DAG task $G=(V,E)$.  
 Let each subjob $v_k \in V$ have the arrival time $a_k$ and finishing time $f_k$ in $S$. 
	We define the envelope of $G$ in $S$ as the collection of arrival and finishing time intervals 
        $[a_{k_1}, f_{k_1}), [a_{k_2}, f_{k_2}), \dots,  [a_{k_p}, f_{k_p})$ for some $p \in \setof{1, \dots, \ell}$ 
        backwards in an iterative manner as follows:
        
\begin{enumerate}
 	\item $k_{i} \neq k_j \in \setof{1,\dots,\ell}$ for all $i\neq j$.
 	\item $v_{k_p}$ is the subjob in $V$ with the maximal finishing time.
 	\item $v_{k_{i-1}}$ is the subjob preceding $v_{k_i}$ with maximal finishing time, for all $i \in \setof{p, p-1, \dots, 2}$.
 	\item $v_{k_1}$ is a source node, i.e., has no predecessor.
 \end{enumerate}
We call  $\pi_e := \setof{v_{k_1} , v_{k_2}, \dots, v_{k_p}}$ the envelope path. 
We note that the definition of an envelope for a DAG job may not be unique if there are subjobs with the same finishing times.
In that case, ties can be broken arbitrarily.
\end{defn}

In the remainder of this subsection, we analyze the response time of a DAG job using 
the previously introduced properties.
Let a schedule $S$ be generated for a single job $J$ according to the policy in Definition~\ref{def:list-fp-dispatching} on $M$ dedicated 
processors using the prioritization described in Definition~\ref{def:parallel-progress-prioritization}. 
For the schedule $S$, the time interval $[a_J, f_J)$ between 
the arrival time and finishing time of said job $J$ is partitioned into times that an envelope sub job 
is executed (busy interval) and times that no envelope subjob is executed (non-busy interval), i.e., 
$I_B := \setof{t \in [a_J, f_J)~|~\text{envelope subjob is executed at $t$}}$ and $I_N := \setof{t \in [a_J, f_J)~|~\text{envelope subjob is not executed at $t$}}$.
In consequence of the fact that $I_B \cap I_N = \emptyset$ and $I_B \cup I_N \equiv [a_J, f_J)$, the response time of DAG job $J$ is given by 
the cumulative amount of time spent in either of these two states.
We require that $\psi$ is at most an $M$- path collection, since 
the number of available processors naturally limits the parallel execution.
For the following analysis, please recall that a subjob $v \in V$ is called pending at time $t$ if $v$ has arrived
and has not yet finished.

Given an envelope path $\pi_e := \setof{v_{k_1} , v_{k_2}, \dots, v_{k_p}}$ for $S$ 
with parallel path prioritization of an $n$-path collection $\psi$, 
we partition each arrival and finishing time interval $[a_{k_i}, f_{k_i})$ 
for each envelope subjob $i \in \setof{1, \dots, p}$ into \emph{busy} $[a_{k_i}, f_{k_i})  \cap I_{B}$ and 
\emph{non-busy} $[a_{k_i}, f_{k_i}) \cap I_{N}$ sub sets. 
Moreover, the \emph{non-busy} partition is further classified into \emph{parallel path non-busy} if $v_{k_i} \in V_s(\psi)$ 
and \emph{non-parallel path non-busy} if $v_{k_i} \in V_s^c(\psi)$. 
The intuition of this approach is to tie the execution of an envelope subjob to the execution 
of subjobs of the path collection $\psi$, which is used in the forthcoming analyses.

\begin{thm}[Preemptive Response Time Bound]
\label{thm:response-time-preemptive}
The response time of a DAG job $J$ with an arbitrary $n$-path collection $\psi = \setof{\pi_{\psi_1}, \dots, \pi_{\psi_n}} \in \mathcal{P}(\Psi(G))$ 
(of $n$ at most $M$) scheduled on $M$ dedicated homogeneous processors using preemptive List-FP scheduling is bounded from above by  
\begin{equation}
\label{eq:response-time-preemptive}
 R_J \leq vol(\pi_*) + \frac{vol(V_s^c(\psi))}{M-n+1} 
\end{equation}
\end{thm}

\begin{proof} 
  By the definition of the envelope (cf. Definition~\ref{def:envelope}), we know that the interval $[a_J, f_J)$ 
  of DAG job's arrival to its finishing time in a concrete preemptive List-FP schedule $S$ can be partitioned 
  into contiguous intervals $[a_{k_1}, f_{k_1} = a_{k_2}), \dots, [f_{k_{n-1}} = a_{k_n}, f_{k_n})$ where $[a_{k_i}, f_{k_i})$ denotes the arrival 
  and finishing time of subjob $v_{k_i}$ for all $i \in \setof{1, \dots, p}$ in the envelope $\pi_e := \setof{v_{k_1} , v_{k_2}, \dots, v_{k_p}}$.
  
  \parlabel{Busy Interval.}
  The amount of busy time is trivially given by $|[a_{k_i}, f_{k_i}) \cap I_B| \leq vol(v_{k_i})$ for each interval, which yields
  that $|[a_J, f_J) \cap I_B| \leq vol(\pi_e) \leq vol(\pi_*)$.
  
  \parlabel{Non-Busy Interval.}
  Since our scheduling policy is work-conserving we have that whenever an envelope subjob $v_{k_i}$ is not executing during $[a_{k_i}, f_{k_i})$ 
  then all $M$ processors must be busy executing non-envelope jobs.
  Since the subjobs $v_{k_i}$ from the envelope can be exclusively either in $V_s(\psi)$ or $V_s^c(\psi)$, we 
  analyze $|I_N \cap  [a_{k_i}, f_{k_i})|$ for both cases individually:
 \begin{itemize}
 \item \textbf{Parallel Path} Let $v_{k_i} \in V_s(\psi)$ and by assumption not execute at time $t$ then
     at most $n-1$ processors execute subjobs from $V_s(\psi)$. That is because
     all subjobs in $V_s(\psi)$ stem from $n$ different paths, which implies that there can never be
     more than $n$-many subjobs from $V_s(\psi)$ pending concurrently.
     Moreover, since $v_{k_i} \in V_s(\psi)$ by assumption and is not executing at $t$ at most $n-1$ processors
     execute subjobs from $V_s(\psi)$. Conversely, we know that at least $M-(n-1)$ processors execute subjobs from
     $V_s^c(\psi)$, since otherwise $v_{k_i}$ would be executed contradicting the case assumption.

\item \textbf{Non-Parallel Path} Let $v_{k_i} \in V_s^c(\psi)$ and by assumption not be executing at time $t$ then 
  it must be that no processor is executing any subjob from $V_s(\psi)$.
   That is because if any of the lower-priority subjobs in $V_s(\psi)$ would be executing,
   then the higher-priority envelope subjob $v_{k_i} \in V_s^c(\psi)$ would be executing as well, 
   which contradicts the case assumption. Conversely, we know that all $M$ processors are exclusively 
   used to execute subjobs from $V_s^c(\psi)$.
\end{itemize}

In conclusion, during each interval $[a_{k_i}, f_{k_i})$, whenever $v_{k_i}$ is not executed then at least $M$ processors execute subjobs from 
$V_s^c(\psi)$ for a parallel path non-busy interval and $M-(n-1)$ processors execute subjobs from $V_s^c(\psi)$ 
for a parallel path busy interval. Since $M-(n-1) < M$ the maximal cumulative amount of time during $[a_J, f_J)$ 
that an envelope subjob is not executed is no more than $vol(V_s^c(\psi))/(M-(n-1))$.
\end{proof}

\begin{cor}[Non-Preemptive Response Time Bound]
\label{cor:response-time-non-preemptive}
 The response time of a DAG job $J$ with an arbitrary $n$-path collection $\psi = \setof{\pi_{\psi_1}, \dots, \pi_{\psi_n}} \in \mathcal{P}(\Psi(G))$ 
(of $n$ {\bf at most $M-1$}) scheduled on $M$ dedicated homogeneous processors using non-preemptive List-FP scheduling is bounded from above by  
\begin{equation}
\label{eq:fed-response}
 R_J \leq vol(\pi_*) + \frac{vol(V_s^c(\psi))}{M-n} 
\end{equation}
\end{cor}

\begin{proof}
 In the non-preemptive case, lower-priority subjobs from $V_s(\psi)$ can interfere 
 with higher-priority subjobs from $V_s^c(\psi)$ if former subjobs arrive earlier than latter.
 Since this inversion only affects the \emph{non-parallel path non-busy interval} analysis, we 
 only analyse this case. Let $v_{k_i} \in V_s^c(\psi)$ and by assumption not be executing at time $t$ then 
 at most $n$ subjobs from $V_s(\psi)$ block the higher-priority subjobs from $V_s^c(\psi)$ from 
 execution due to non-preemption. Conversely, we know that at least $M-n$ processors are 
 used to execute subjobs from $V_s^c(\psi)$. By similar arguments as in the proof of Theorem~\ref{thm:response-time-preemptive} 
 and since $M-n < M-(n-1)$ the maximal cumulative amount of time during $[a_J, f_J)$ 
that an envelope subjob is not executed is no more than $vol(V_s^c(\psi))/(M-n))$ for $M > n$ and the amount of time 
envelope subjobs are executed are $vol(\pi_e) \leq vol(\pi_*)$, which concludes the proof.
\end{proof}

\parlabel{Sustainability of our response time analysis.} Many multiprocessor hard real-time scheduling algorithms and schedulability 
analyses presented in the literature are not sustainable, which means that they suffer from timing anomalies. 
These anomalies describe the counter-intuitive phenomena 
that a job that was verified to always meet its deadline can miss its deadline by augmenting resources, 
e.g., to execute the job on more processors or to decrease the execution-time (early completion).
In Corollary~\ref{cor:early-completion}, we show that our response time bound is sustainable with respect 
to the number of processors and the subjob execution-time. This is a beneficial property
in dynamic environments, where available processors and execution times vary, and ultimately  
simplifies implementation efforts in real systems. 

\begin{cor}[Sustainability]
\label{cor:early-completion}
 The response time bounds in Theorem~\ref{thm:response-time-preemptive} and 
 Corollary~\ref{cor:response-time-non-preemptive}
 holds true for a DAG job with $G=(V,E)$ even if  any subjob $v \in V$ 
 completes before its worst-case execution time or if the number of processors is increased.
\end{cor}

\begin{proof}
 This comes directly from the observation that the volume of the envelope path $vol(\pi_e)$ 
 as well as the length of \emph{parallel path non-busy} and 
 \emph{non-parallel path non-busy} intervals can only decrease if the worst-case execution time of any subjob decreases 
 or the number of processors is increased.
 Since the response time is upper-bounded by the sum of times that an envelope job is executed and 
 the sum of times that no envelope job is executed, the corollary is proved.
\end{proof}

\subsection{Parallel Path Progression Schedule Example \& Discussion}
\label{sec:schedule-example-and-discussion}

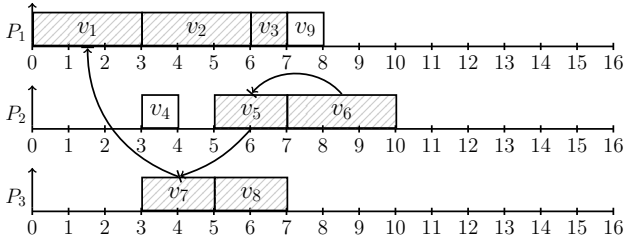
\begin{figure}[tb]
\centering
\resizebox{.95\linewidth}{!}{\begin{tikzpicture}[]
\def\ux{1.7cm}\def\uy{1.20cm}
    \def\dx{0.88cm}
        \tikzset{
          job/.style={fill=white, text=black, very thick, minimum height=8mm, draw},
          interference/.style={pattern=north east lines, text=black, minimum height=7mm},
          collection/.style={pattern=north east lines, pattern color=black!20, text=black, minimum height=8mm, very thick},
          nonservice/.style={fill=black!80, very thick, draw, text=black, minimum height=8mm},
	  envelope/.style={fill=gray!40,  text=black, very thick, minimum height=8mm, draw },
          prec/.style={color=black, very thick}
        }

	\begin{scope}[shift={(0, 6)}] 
           \foreach \x in {0, 1,...,16} 
           \draw[-, very thick](\x*\dx, 0.1*\dx) -- (\x*\dx, -0.1*\dx)node[below] {\Large $\x$};
           \draw[-, very thick] (0,0) node[anchor=south east] {{\Large $P_1$}}-- coordinate (xaxis) (16*\dx,0);
           \draw[->, very thick](0*\dx, 0.1*\dx) -- (0*\dx, 1.2*\dx)node[below] {};
           
           \node[collection, minimum width=3*\dx, anchor=south west,draw] at (0*\dx,0) (v1) {\LARGE $v_1$};
           \node[collection, minimum width=3*\dx, anchor=south west,draw] at (3*\dx,0) (v2) {\LARGE $v_2$};
           \node[collection, minimum width=1*\dx, anchor=south west,draw] at (6*\dx,0) (v3) {\LARGE $v_{3}$};
           \node[job, minimum width=1*\dx, anchor=south west,draw] at (7*\dx,0) (v9) {\LARGE $v_9$};
        \end{scope}

        \begin{scope}[shift={(0, 4)}] 
           \foreach \x in {0, 1,...,16} 
           \draw[-, very thick](\x*\dx, 0.1*\dx) -- (\x*\dx, -0.1*\dx)node[below] {\Large $\x$};
           \draw[-, very thick] (0,0) node[anchor=south east] {{\Large $P_2$}}-- coordinate (xaxis) (16*\dx,0);

            \draw[->, very thick](0*\dx, 0.1*\dx) -- (0*\dx, 1.2*\dx)node[below] {};
            \node[job, minimum width=1*\dx, anchor=south west,draw] at (3*\dx,0) (v4) {\LARGE $v_4$};
            \node[collection, minimum width=2*\dx, anchor=south west,draw] at (5*\dx,0) (v5) {\LARGE $v_5$};
            \node[collection, minimum width=3*\dx, anchor=south west,draw] at (7*\dx,0) (v6) {\LARGE $v_{6}$};
        \end{scope}

        \begin{scope}[shift={(0, 2)}] 
           \foreach \x in {0, 1,...,16} 
           \draw[-, very thick](\x*\dx, 0.1*\dx) -- (\x*\dx, -0.1*\dx)node[below] {\Large $\x$};
           \draw[-, very thick] (0,0) node[anchor=south east] {{\Large $P_3$}}-- coordinate (xaxis) (16*\dx,0);

            \draw[->, very thick](0*\dx, 0.1*\dx) -- (0*\dx, 1.2*\dx)node[below] {};
            \node[collection, minimum width=2*\dx, anchor=south west,draw] at (3*\dx,0) (v7) {\LARGE $v_7$};
            \node[collection, minimum width=2*\dx, anchor=south west,draw] at (5*\dx,0) (v8) {\LARGE $v_8$};
        \end{scope}

         \draw[->|,prec] (v6.north) to [bend right=50] (v5.north);
         \draw[->|,prec] (v5.south) to [bend left=15] (v7.north);
         \draw[->|,prec] (v7.north) to [bend left=35] (v1.south);
         
    \end{tikzpicture}

}
\caption{An exemplary List-FP schedule for the DAG in Figure~\ref{fig:example-dag-plain} 
is shown in Figure~\ref{fig:parallel-path-progression-example} on $3$ processors. A $3$-path collection 
$\psi := \setof{\setof{v_1, v_7, v_5, v_6}, \setof{v_1, v_7, v_8}, \setof{v_1, v_2,  v_3}}$ is chosen, which results in 
the corresponding set of subjobs $V_s(\psi) = \setof{v_1, v_2, v_3, v_5, v_6, v_7, v_8}$ (hatched) and $V_s^c(\psi) = \setof{v_4, v_9}$ (blank).}
\label{fig:parallel-path-progression-example}
\end{figure}

In this subsection, we exemplify the introduced parallel path progression scheduling properties. We also 
demonstrate the resulting response time analysis improvements of our approach 
against federated scheduling, which can be seen as a pessimistic special case of our approach 
when only one path is considered, i.e., $|\psi| = n=1$. 

A parallel path progression schedule for the DAG in Figure~\ref{fig:example-dag-plain} on $3$ processors 
is shown in Figure~\ref{fig:parallel-path-progression-example}. A $3$-path collection 
\mbox{$\psi := \setof{\setof{v_1, v_7, v_5, v_6}, \setof{v_1, v_7, v_8}, \setof{v_1, v_2,  v_3}}$} is chosen, which results in 
the corresponding set of subjobs $V_s(\psi) = \setof{v_1, v_2, v_3, v_5, v_6, v_7, v_8}$ of the parallel progress paths 
and $V_s^c(\psi) = \setof{v_4, v_9}$. The subjobs in $V_s(\psi)$ are emphasized by the hatchings, whereas subjobs in $V_s^c(\psi)$ 
are left blank.
In the given example, no preemption is required since at no point in time 
are there more subjobs pending than processors available. Subsequently, the pending intervals coincide with 
the execution intervals of each subjob.

By inspection of the schedule, we observe that 
each time $t \in [0,10)$ is parallel path progressive, since at each time all pending subjobs of 
$V_s(\psi)$ are executed. Respectively, the cumulative non-parallel path progressive set length is $0$ 
by visual inspection.

Starting from the latest finishing subjob $v_6$ the envelope of the 
schedule is emphasized by the arrows.
All subjobs of the envelope path $\pi_e := \setof{v_6, v_5, v_7, v_1}$ are subjobs from $V_s(\psi)$ and thus 
the cumulative parallel path progressive set length is exactly given by the volume of the longest path, which is $10$ 
in the shown example. 

In contrast to the exact result, the analytic bound as stated in Eq.~\eqref{eq:response-time-preemptive} 
bounds the cumulative response time by $10+2=2$. 
The improvements compared to federated scheduling~\cite{Li:ECRTS14} are directly observable. 
In federated scheduling, only the longest path 
in the DAG is considered, which results in a more pessimistic response time upper-bound of $10+8/3=12.7$.

\section{Path Collection Algorithm}
\label{sec:n-path-collection-approximation-algorithm}

In our approach, any $n$-path collection, where $n$ is at most the number of 
dedicated processors $M$ can potentially be chosen. We are however interested in the minimal 
achievable worst-case response time of a single DAG job (i.e., the minimum makespan of the DAG job) 
as stated in Eq.~\eqref{eq:fed-response}.
While it is obvious that in the case that $n=1$,  the volume of the set $V^c_s(\psi)$ 
is minimal if the longest path is chosen, finding a \emph{makespan optimal} $n$-path collection 
for $n > 1$ is non-trivial to the best of our knowledge.

For a given DAG $G=(V,E)$ however, it is possible to determine the minimal number of paths 
to cover all vertexes in $V$ in polynomial time as described by Dilworth
~\cite{Dilworth1990} and K\"onig-Egevary~\cite{DEMING}.
The minimal required number of paths to cover a DAG $G$ is commonly denoted as 
the \emph{width} of $G$. The paths that cover the complete DAG can be calculated 
in polynomial time using a reduction to the maximal matching problem in bipartite graphs.
We use $\textsc{PathCover}$ to refer to this algorithm and point to the literature 
for a detailed discussion of that algorithm.

Another related algorithm is the \textsc{Weighted Maximum Coverage}~\cite{Nemhauser1978AnAO} problem.
Hereinafter, we map the problem of finding of an $n$-path collection $\psi$ for a
DAG $G$ that maximizes $vol(V_s(\psi))$ to this problem as follows:
\begin{itemize}
\item \textbf{Input}:
  A problem instance $I$ of the \textsc{Weighted Maximum Coverage} problem is given by a collection of sets 
  $S := \setof{S_1, \dots, S_m}$, a weight function $\omega$, and a natural number $k$.
  Each set $S_i \subseteq U$ is a subset from some universe $U$ for each $i \in \setof{1, \dots, m}$ 
  and each element $s \in S_i$ is associated with a weight as given by the function $\omega(s)$.
 
\item \textbf{Objective}: For a given problem instance $I$, the
  objective is to find a subset $S' \subseteq S$ such that
  $|S'| \leq k$ and
  $\sum_{s \in \setof{\cup \setof{S_i \in S'}}} w(s)$ is maximized.
\end{itemize}

It was shown by Nemhauser et al.~\cite{Nemhauser1978AnAO} that any
polynomial time approximation algorithm of the \textsc{Weighted
  Maximum Coverage} problem has an asymptotic approximation ratio with
respect to an optimal solution that is lower-bounded by $1-1/e$ unless
$P=NP$, where $e$ is Euler's number.  This approximation ratio can be
achieved by a greedy strategy that always chooses the set which
contains the largest weights of not yet chosen elements. 
Despite \textsc{Weighted Maximum Coverage} and our problem not being equivalent, 
we use the approximation strategy for the $n$-path Collection Approximation 
in Algorithm~\ref{alg:n-path-collection-approximation}.

\begin{algorithm}[tb]
\caption{$n$-Path Collection Approximation (nPCA)}
\label{alg:n-path-collection-approximation}
\begin{algorithmic}[1]
\Require {DAG $G=(V,E)$, No. CPU $M$, WCET $vol$.}
\Ensure An approximately optimal $n$-path collection $\psi$
\State{$\psi^* \leftarrow$ $\small{\textsc{PathCover}}(G) :=   \setof{\pi_{\psi_1}, \dots, \pi_{\psi_{w}}}$} \label{alg:line:opt-begin}
\If{$w \leq M$}
\State{\textbf{return} $(\psi^*, w)$;} \label{alg:line:opt-end}
\EndIf
\State{create $\psi_0 \leftarrow \emptyset$;} \label{alg:line:approx-begin}
\State{$z \leftarrow \infty$;}
\State{$C_i \leftarrow vol(V)$;}
\For{\textbf{each} $n \in \setof{1, \dots, M}$}\label{alg:line:approx-outer}
\State{create $\psi_n \leftarrow \psi_{n-1}$;}
 \State{$\pi_n^* \leftarrow$ use $\small{\textsc{DFS}(G)}$ to search the max. volume path;} \label{alg:line:dfs}
   \State {$\psi_{n} \leftarrow \psi_{n} \cup \pi_n^*$;}
    \For{\textbf{each} $v \in \pi_n^*$}
     \State {update $vol(v)$ to $0$;}
     \EndFor
     \State{$z' \leftarrow (C_i - vol(\psi_n))/(M-(n-1))$;} \label{alg:line:compare}
     \If{$z' < z$} \label{alg:line:min-compare}
     \State{solution $(\psi^*, n^*) \leftarrow (\psi_n, n)$;}
     \State{update $z \leftarrow z'$;}
    \EndIf
 \EndFor
\State{\textbf{return} solution $(\psi^*, n^*)$;} \label{alg:line:approx-end}
\end{algorithmic}
\end{algorithm}

\parlabel{$n$-Path Collection Approximation Algorithm.}
From line~\ref{alg:line:opt-begin} to line~\ref{alg:line:opt-end} in 
Algorithm~\ref{alg:n-path-collection-approximation}, the minimal number $w$ of 
paths and the respective paths to cover the complete DAG is computed. 
If the number of processors $M$ is sufficient to allow the parallel execution of all $w$ paths, 
i.e., $w \leq M$ then those paths are chosen for the path collection.

In the other case, from line~\ref{alg:line:approx-begin} to \ref{alg:line:approx-end} an 
approximate solution is computed. 
In each iteration $n \in \setof{1, \dots, M}$, the longest path $\pi_n^*$ with 
respect to the current iteration's volume function is chosen. 
After the path is chosen, all volumes of that path's subjobs are set to $0$ to indicate 
that the subjobs have already been covered.
By this strategy, we always choose the path,  which contains the largest amount of volume of not yet chosen subjobs 
in each iteration.
Moreover, in each $n$-th iteration, it is probed in line~\ref{alg:line:min-compare} if 
the solution $\psi_n$ strictly improves the prior solution $\psi_{n-1}$ 
with one path less (line~\ref{alg:line:compare}) .
At the end of the $M$-th iteration, an $n^*$-path collection $\psi^*$ is found that yields 
formal guarantees as stated in Theorem~\ref{thm:n-path-collection-approximation}.

The maximal bipartite matching can be obtained in \mbox{$\mathcal{O}(|V| \cdot |E|)$} 
using the Ford-Fulkerson algorithm.
The time-complexity of $n$PCA is dominated 
by the for-loop and the depth-first search (DFS) in line~\ref{alg:line:dfs} that is invoked in each of the iterations, 
resulting in \mbox{$\mathcal{O}(M \cdot |V|\cdot |E|)$}. 

\begin{thm}[nPCA]
\label{thm:n-path-collection-approximation}
The worst-case response time of a DAG job $J$ on $M$ 
dedicated processors using parallel path 
progression scheduling and for which the $n^*$-path collection 
is calculated according to Algorithm~\ref{alg:n-path-collection-approximation}, 
is at most
\begin{equation} 
R_{opt} \cdot 
\begin{cases}
  1+ \frac{M}{M-n^*+1} \cdot \left(1-\frac{1}{w}\right)^{n^*} \leq 2-\frac{1}{w}& w > M \geq n^* \\
1 & \, M \geq w
\end{cases}
\end{equation}
where $w$ denotes the width of the DAG.
\end{thm}

\begin{proof}
 We prove this theorem for the cases $M \geq w$ and $ M < w$ individually.
 \begin{itemize}
   \item In the first case let $M \geq w$. By definition of the width $w$ of a DAG $G=(V,E)$, each $v \in V$ is covered by at least 
 one of those $w$ many paths. Obviously, in the case that $w \leq M$, the graph's total workload 
 can be covered by those $w$ paths. 
 The optimal paths are returned in line~\ref{alg:line:opt-end}.
 In consequence, the response-time bound is given by
     \begin{equation}
       R_J \leq vol(\pi_*) + \frac{0}{M-(w-1)} \leq R_{opt}
      \end{equation}
\item In the other case let $w > M \geq n \geq 1$, i.e.,  $V$ can not be fully covered by any $n$ paths 
  and thus the graph's total workload can be neither.
  Thus starting from line~\ref{alg:line:approx-outer}, we approximate an optimal $n$-path collection 
  where $n$ is at most $M$. In the remainder of this proof, we will use $C_i$ to denote $vol(V)$ 
  with the initial $vol$ function, i.e., before any updates have occurred during the execution of the 
  algorithm.
 In the $n$-th iteration of line~\ref{alg:line:approx-outer} 
 it must be that there exists a path $\pi$ such that $vol(\pi)$ covers at least $1/w$ of 
 $C_i-vol(\psi_{n-1})$. Since by definition $vol(\pi_n^*) \geq vol(\pi)$ holds, we have that 
 \begin{equation}
  \label{eq:apx-lower-bound}
   vol(\pi_n^*) \geq \frac{C_i-vol(\psi_{n-1})}{w}
  \end{equation} where $\psi_0 := \emptyset$ and $vol(\psi_0) = 0$. We prove by induction that for each $n \in \setof{n \in \mathbb{N}~|~w > M \geq n \geq 1}$
 \begin{equation} 
   \label{eq:induction-hypothesis}
   C_i-vol(\psi_{n}) \leq \left(1-\frac{1}{w}\right)^{n} \cdot C_i
\end{equation} holds.

\parlabel{Induction Hypothesis.}
For $n=1$, Eq.~\eqref{eq:induction-hypothesis} reduces to $w \cdot vol(\psi_1) \geq C_i$, 
which holds true since we know that there exists some $w < \infty$ such that $vol(\pi_{\psi_1} \cup \dots \cup \pi_{\psi_{w}}) = C_i 
\leq \sum_{i=1}^w vol(\pi_{\psi_{i}}) \leq w \cdot vol(\psi_{1})$, since $\psi_1$ only contains the longest path 
in $G$.

\parlabel{Induction Step.}
In the induction step $n \rightarrow n+1$, we have that
\begin{equation}
\label{eq:induction-init}
  C_i-vol(\psi_{n+1}) = C_i-(vol(\psi_n) + vol(\pi^*_{n+1})) 
\end{equation}
Using Eq.~\eqref{eq:apx-lower-bound}, we conclude that
\begin{align}
  \text{Eq.}~\eqref{eq:induction-init} &\leq C_i-vol(\psi_n) - \frac{C_i-vol(\psi_{n})}{w} \\
&\leq  (C_i-vol(\psi_n)) \cdot \left(1-\frac{1}{w}\right) \\
&\leq \left(1-\frac{1}{w}\right)^{n} \cdot C_i \cdot \left(1-\frac{1}{w}\right)
\end{align}

\parlabel{Conclusion.}
In conclusion for the considered case, Eq.~\eqref{eq:induction-hypothesis} 
yields that after the $n$-th iteration of $n$PCA,
\begin{equation}
 (1- \left(1-\frac{1}{w}\right)^{n}) \cdot C_i \leq vol(\psi_{n})
\end{equation} holds and thus the maximum response-time using the computed 
$n$-path collection $\psi_n$ is at most
\begin{equation}
\label{eq:response-time-bound}
     R_J  \leq vol(\pi_*) + \frac{C_i}{M} \cdot \frac{M}{M-n+1} \cdot \left(1-\frac{1}{w}\right)^{n}
\end{equation}

Due to  the fact that $R_{opt} \geq \max \setof{vol(\pi_*), C_i/M}$ and the minimal 
response time solution $(\psi^*, n^*)$ returned by $n$PCA 
in line~\ref{alg:line:approx-end} using Eq.~\eqref{eq:response-time-bound} 
we have that
\begin{align}
R_J &\leq R_{opt} \cdot \min_{n\geq 1} \setof{1+ \frac{M}{M-n+1} \cdot \left(1-\frac{1}{w}\right)^{n}} \\
& \leq R_{opt} \cdot \min_{n\geq 1} \setof{1+ n \cdot \left(1-\frac{1}{w}\right)^{n}} \\
& \leq R_{opt} \cdot  \left(2-\frac{1}{w}\right) \label{eq:bound-second}
\end{align} 

Additionally, we have that
\begin{equation}
\label{eq:bound-first}
R_J \leq R_{opt} \cdot  \left(1+ \frac{M}{M-n^*+1} \cdot \left(1-\frac{1}{w}\right)^{n^*}\right)
\end{equation} 
Finally, we have 
$R_J \leq Eq.~\eqref{eq:bound-first} \leq Eq.~\eqref{eq:bound-second}$, which concludes the proof.
 \end{itemize}
\end{proof}

\section{Hierarchical Scheduling}
\label{sec:hierarchical-scheduling}

We extend the properties of parallel path progression to  
a system with inter-task interference using a hierarchical scheduling 
approach. That is, the scheduling problem is separated into different scheduling levels.
On the lowest level, reservation systems are scheduled on the physical processors 
by some scheduling policy that is compliant with the model of the reservation system.
On the higher level, the workload, i.e., the subjobs of a DAG job, is executed by the reservations 
in a temporally and spatially isolated environment.
The isolation property allows to analyse each DAG job's response time without 
inter-task interference.
Most importantly, reservation systems can be co-scheduled with other tasks on the same set of physical 
processors using existing response time analyses.

We propose and discuss two reservation schemes, namely a \emph{gang reservation system} in
Section~\ref{sec:gang-reservation} and an \emph{ordinary reservation system} in Section~\ref{sec:ordinary-reservation}, 
and provide resource provisioning rules and response time analyses.
The hierarchical scheduling problem consists of two interconnected problems:
\begin{itemize}
 \item Service provisioning of the respective \emph{gang-}reservation or  \emph{ordinary-}reservation systems 
   such that a DAG job can finish within the provided service.
 \item Verification of the schedulability of the provisioned reservation systems by 
 any existing analyses 
   that support the respective task models, e.g., sporadic arbitrary-deadline gang tasks 
   or sporadic arbitrary-deadline ordinary sequential tasks.
\end{itemize}
For the remainder of this section, we assume the existence of a feasible
schedule upon
$M$ identical multiprocessors for the studied reservation system model 
and focus on the service provisioning problem.
We assume that a reservation system satisfies the following four properties regardless 
of the specific reservation model.

\parlabel{Property 1 (Parallel Service).}
The reservation systems release $m$ parallel reservations such that 
at each time, during which the reservation system promises service, at most $m$ reservations 
can provide service concurrently.

\parlabel{Property 2 (Association of Service).}
An instance of the reservation system serves exactly one DAG job of a DAG task.
This means that an instance of the $m$ parallel reservations that serve the \mbox{$\ell$-th} job $J_i^{\ell}$ 
of DAG task $\tau_i$ all arrive synchronous at time $a_i^{\ell}$ and the deadline is given by $d_{i}^{\ell}$.
Note that if the next DAG job arrives before the previous one is finished, a 
new instance (job) of reservations is released. 
This allows to directly deal with arbitrary-deadlines without further considerations.

\parlabel{Property 3 (Sustained Service).}
The service of a reservation is provided whenever the reservation system is 
scheduled, irrespective of whether there are insufficient number of pending subjobs to be served 
at that time. 

\parlabel{Property 4 (Internal Dispatching).}
The reservation system's internal dispatching of each subjob $v \in V$  of a DAG job $G=(V,E)$ 
with parallel path-progression prioritization on $m$ reservations follows List-FP in Definition.~\ref{def:list-fp-dispatching}
with only one difference: At any time $t$ the, let $m(t)$ denote the concurrently scheduled reservations then the $m(t)$ 
highest-priority pending subjobs are executed on the reservations and a lower-priority subjob is preempted if necessary.

\subsection{Gang Reservation System}
\label{sec:gang-reservation}

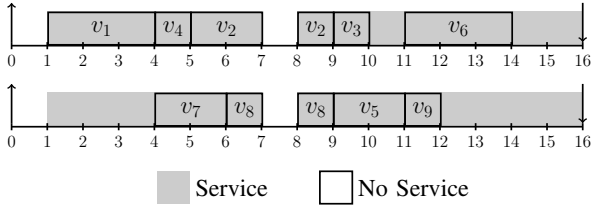
\begin{figure}[tb]
\centering
\resizebox{.9\linewidth}{!}{\begin{tikzpicture}[]
\def\ux{1.7cm}\def\uy{1.20cm}
    \def\dx{0.88cm}
        \tikzset{
          job/.style={fill=white, text=black, very thick, minimum height=8mm},
          interference/.style={pattern=north east lines, text=black, minimum height=7mm},
          collection/.style={ text=black, minimum height=8mm, very thick},
          nonservice/.style={fill=black!30, very thick, draw, text=black, minimum height=8mm},
          service/.style={fill=black!20, text=black, minimum height=8.5mm},
	  envelope/.style={fill=gray!40,  text=black, very thick, minimum height=8mm, draw },
          prec/.style={color=black, very thick}
        }

        \begin{scope}[shift={(0, 3.5)}] 
           \node[service, minimum width=6*\dx, anchor=south west] at (1*\dx,0) (service) {};
           \node[service, minimum width=8*\dx, anchor=south west] at (8*\dx,0) (service) {};
          
           \foreach \x in {0, 1,...,16} 
           \draw[-, very thick](\x*\dx, 0.1*\dx) -- (\x*\dx, -0.1*\dx)node[below] {\large $\x$};
           \draw[-, very thick] (0,0) node[anchor=south east] {}-- coordinate (xaxis) (16*\dx,0);
           \draw[->, very thick](0*\dx, 0.1*\dx) -- (0*\dx, 1.2*\dx)node[below] {};
           \draw[<-, very thick](16*\dx, 0.1*\dx) -- (16*\dx, 1.2*\dx)node[below] {};

           \node[collection, minimum width=3*\dx, anchor=south west,draw] at (1*\dx,0) (v1) {\LARGE $v_1$};
           \node[collection, minimum width=1*\dx, anchor=south west,draw] at (4*\dx,0) (v1) {\LARGE $v_4$};
           \node[collection, minimum width=2*\dx, anchor=south west,draw] at (5*\dx,0) (v1) {\LARGE $v_2$};
           \node[collection, minimum width=1*\dx, anchor=south west,draw] at (8*\dx,0) (v1) {\LARGE $v_2$};
           \node[collection, minimum width=1*\dx, anchor=south west,draw] at (9*\dx,0) (v1) {\LARGE $v_3$};
           \node[collection, minimum width=3*\dx, anchor=south west,draw] at (11*\dx,0) (v1) {\LARGE $v_6$};
           
        \end{scope}

        \begin{scope}[shift={(0, 1.5)}] 
          
           \node[service, minimum width=6*\dx, anchor=south west] at (1*\dx,0) (service) {};
           \node[service, minimum width=8*\dx, anchor=south west] at (8*\dx,0) (service) {};
          
           \foreach \x in {0, 1,...,16} 
           \draw[-, very thick](\x*\dx, 0.1*\dx) -- (\x*\dx, -0.1*\dx)node[below] {\large $\x$};
           \draw[-, very thick] (0,0) node[anchor=south east] {}-- coordinate (xaxis) (16*\dx,0);
           \draw[->, very thick](0*\dx, 0.1*\dx) -- (0*\dx, 1.2*\dx)node[below] {};
           \draw[<-, very thick](16*\dx, 0.1*\dx) -- (16*\dx, 1.2*\dx)node[below] {};

           \node[collection, minimum width=2*\dx, anchor=south west,draw] at (4*\dx,0) (v1) {\LARGE $v_7$};
           \node[collection, minimum width=1*\dx, anchor=south west,draw] at (6*\dx,0) (v1) {\LARGE $v_8$};
           \node[collection, minimum width=1*\dx, anchor=south west,draw] at (8*\dx,0) (v1) {\LARGE $v_8$};
           \node[collection, minimum width=2*\dx, anchor=south west,draw] at (9*\dx,0) (v1) {\LARGE $v_5$};
           \node[collection, minimum width=1*\dx, anchor=south west,draw] at (11*\dx,0) (v1) {\LARGE $v_9$};
        \end{scope}

         \begin{scope}[shift={(0, 0)}] 
                \node[service, minimum height=8mm, minimum width=8mm, rectangle, label=0:\LARGE Service] at (4, 0) {};
                \node[job, draw, minimum height=8mm, minimum width=8mm, rectangle, label=0:\LARGE No Service] at (8,0) {};
              \end{scope} 
         
    \end{tikzpicture}

}
\caption{Exemplary schedule for a $2$-gang reservation system denoted as $\tau_1$ and $\tau_2$ 
  with a $2$- path collection 
 of the DAG illustrated in Figure~\ref{fig:example-dag-plain} with a deadline of $16$ time units.}
\label{fig:example-schedule-gang-reservation}
\end{figure}

In gang scheduling, a set of threads is grouped together into a so-called gang with
the additional constraint that all threads of a gang must be co-scheduled at the same time
on available processors. It has been demonstrated that gang-based parallel computing can often
improve the performance~\cite{Feitelson92, Jette97, Wasly19}.
Due to its performance benefits, the gang model is supported
by many parallel computing standards, e.g., \emph{MPI}, \emph{OpenMP}, 
\emph{Open ACC}, or \emph{GPU computing}. 
Motivated by the practical benefits and the conceptual fit of parallel path progressions 
in our approach and the gang execution model, we propose an $m$-Gang reservation 
system as follows.

\begin{defn}[$m$-Gang Reservation System]
  \label{def:gang-reservation-system}
  A sporadic arbitrary-deadline $m_i$-gang reservation system $\mathcal{G}$ 
 that serves a sporadic arbitrary-deadline DAG task $\tau_i := (G_i, D_i, T_i)$ is 
 defined by the tuple $\mathcal{G}_i := (m_i, E_i, D_i, T_i)$ such that $m_i \cdot E_i$ 
 service is provided during the arrival- and deadline interval with the gang scheduling constraint 
 that all reservations must be co-scheduled at the same time.
\end{defn}

Hence, the provisioning problem of $\mathcal{G}_i$ for a DAG task $\tau_i$ 
is to find $m_i$ and $E_i$ such that 
given the properties 1-4 and the gang scheduling constraint, each DAG job can complete within one of the 
$m_i$ reservations before its absolute deadline.

\begin{thm}[Gang Reservation Provisioning]
\label{thm:response-time-federated-bi-partition-ops}
Each job $J_i^{\ell}$ of a sporadic arbitrary-deadline DAG task $\tau_i := (G_i, D_i, T_i)$ can complete 
its total volume within its respective gang reservation instance of the $m_i$ parallel reservations of size $E_i$ 
before its absolute deadline if
\begin{equation}
  \label{eq:gang-res}
  vol(\pi_{*}) + \frac{vol(V_s^c(\psi))}{m_i-n+1} \leq E_{i}  \leq D_i
  \end{equation} holds for any $n$- path collection $\psi$ of at most $m_i$ 
and the gang reservation system $\mathcal{G}_i$ is verified to be schedulable, i.e., 
able to provide all service before the absolute deadline. 
\end{thm}

\begin{proof}
 Since in an $m$-gang all reservations provide service simultaneously, the arrival 
 and finishing time window $[a_J, f_J]$ of any DAG job $J$ of task $\tau_i$ can be partitioned into 
 \emph{busy}, \emph{non-busy}, and \emph{non-service} intervals 
 in which none of the $m_i$ gang reservations are scheduled and thus provide no service. 
 In analogy to previous proofs, the response time is no more than the cumulative length of these 
 sets, where the cumulative length of non-service times is upper-bounded by $D_i-E_i$ given  
 the assumption that the $m_i$ gang is schedulable.
 In consequence, if Eq.~(\ref{eq:gang-res}) holds, then
 \begin{equation}
   \label{eq:gang-res-result}
   R_J \leq vol(\pi_{*})+\frac{vol(V_s^c(\psi))}{m_i-n+1}+D_{i}- E_{i} \leq D_{i}
 \end{equation} which  concludes the proof.
\end{proof}

For illustration, an example of a feasibly provisioned gang reservation system 
is shown in Figure~\ref{fig:example-schedule-gang-reservation} for the DAG in 
Figure~\ref{fig:example-dag-plain} with  
longest path volume $10$, total volume $18$, 
and relative deadline $16$. Using Eq.~\eqref{eq:gang-res} for a $2$-gang and 
the \mbox{$2$-} path collection $V_s(\psi) = \setof{v_1, v_7, v_5, v_6, v_2, v_3}$ yields 
$10 + (18-14)/(2-2+1) = 14  \leq 16$. Figure~\ref{fig:example-schedule-gang-reservation} 
shows a feasible schedule of this provisioned $2$-gang system.

\begin{algorithm}[tb]
\caption{Feasible Minimal Waste $m_i$-Gang}
\label{alg:gang-reservation-design}
\begin{algorithmic}[1]
\Require {DAG Task $\tau_i := (G_i, D_i, T_i)$, No. CPU $M$, $vol_i$.}
\Ensure {Minimal Waste $m_i$-Gang $\mathcal{G}_i$ for DAG task $\tau_i$ and $n$.}
\State {$\xi \leftarrow [0, \dots, 0]$;}
\For{\textbf{each} $n \in \setof{1, \dots, M}$}
  \State {search the longest path $\pi_*$ in $G_i$ with respect to $vol_i$;}
  \If{$vol(\pi_*) > 0$} \Comment{Some subjobs not yet visited}
  \State{$\xi[n] \leftarrow \xi[n-1] + vol_i(\pi_*)$;}
     \State {set $vol_i(v) := 0$ for each $v$ in $\pi_*$;}
   \Else
   \State{break;}
\EndIf
\EndFor
\State{$W_{\sigma} \leftarrow [\emptyset, \infty];$} \Comment{$((m_i,n), waste(m_i, n))$}
\For{\textbf{each} $m_i \in \setof{1, \dots, M}$}
 \For{\textbf{each} $n \in \setof{1, \dots, m_i}$}
 \If{$E_i \leftarrow$ using Eq.~\eqref{eq:gang-res} with $(m_i, n) > D_i$}
 \State{continue;}
\EndIf
   \State{
     $
       W \leftarrow m \cdot \xi[1]+\frac{m_i}{m_i-n+1} (C_i-\xi[n])- C_i
     $;}
  \If{$W < W_{\sigma}[2]$;}
    \State{$W_{\sigma}[1] \leftarrow (m_i, n)$ and $W_{\sigma}[2] \leftarrow W$;}
  \EndIf
 \EndFor
\EndFor

\State{\textbf{return} $\mathcal{G}_i = (m_i, E_i, D_i, T_i)$ and $n$;}
\end{algorithmic}
\end{algorithm}

Finding the best provisioning for specific gang reservation systems depends 
on the concrete schedulability problem at hand and the other tasks that are to be co-scheduled. 
Hence, it may be beneficial to trade decreased budgets for increased gang size in some concrete scenarios.
Determining such specific provisions is  beyond the scope of this work.
In general, due to the gang restriction, increasing the number of reservations $m_i$ 
beyond the number of processors $M$ that the reservations are going to be executed upon 
is not possible. In a more generic optimization attempt, we seek to find a provisioning that minimizes the unused gang service (waste), which 
is described as $m_i \cdot E_i - C_i$. 
Due to the fact that $n \in \setof{1, \dots, m_i}$  and 
$m_i \in \setof{1, \dots, M}$, an exhaustive search can be applied to 
find the feasible tuple $((m_i, E_i, D_i, T_i), n)$ with $E_i \leq D_i$  that approximately minimize the waste objective 
\begin{equation}
  m_i \cdot vol(\pi_{*})+m_i \cdot \frac{C_i-vol(V_s(\psi))}{m_i-n+1} - C_i
\end{equation} as shown in Algorithm~\ref{alg:gang-reservation-design}.

\subsection{Ordinary Reservation System}
\label{sec:ordinary-reservation}

Despite the practical benefits of gang scheduling, the analytic schedulability 
due to the co-scheduling constraint is reduced compared to an equivalent ordinary reservation system 
without such constraints.

\begin{thm}[$m$-Ordinary Reservation Provisioning]
\label{thm:ordinary-reservation-provisioning}
 Each job $J_i^{\ell}$ of a sporadic arbitrary-deadline DAG task \mbox{$\tau_i := (G_i, D_i, T_i)$} can complete 
 within its respective ordinary reservation instance $\mathcal{O}_i$ before its absolute deadline if 
 firstly 
 \begin{equation}
   \label{eq:res-theorem}
  vol(\pi_{*}) + \frac{vol(V_s^c(\psi)) + (n-1) \cdot D_{i}}{m_i-n+1} \leq \frac{\sum_{p=1}^m E_i^p}{m_i-n+1} 
\end{equation}
 holds for any $n$-path collection $\psi$ where $n$ is at most $m_i$ and $E_i^p \geq vol(\pi_*)$ for $p \in \setof{1, \dots, m_i}$;
 and secondly, the ordinary reservation system $\mathcal{O}_i$ is verified to be schedulable, i.e., 
 each individual reservation is able to provide all service before the absolute deadline. 
\end{thm}

\begin{proof}

Let $S$  be a schedule of $m_i$ ordinary reservations that 
are feasibly schedulable.
That is, each individual reservation is guaranteed to provide $E_i^p$ service 
to the DAG job $J := J_i^{\ell}$ during the interval $[a_J, a_J+D_i)$ for $p \in \setof{1,\dots, m_i}$.
Let $\psi := \setof{\pi_{\psi_1}, \dots, \pi_{\psi_n}}$ denote the $n$-path collection that composes the set of subjobs $V_s(\psi)$. 
Let $m_i(t) \in \setof{0, 1, \dots, m_i}$ denote the number of reservations that provide service 
  at time $t$ and let  $h(t) \in \setof{0, \dots, m_i(t)}$ denote the number of reservations providing service to  subjobs in $V_s^c(\psi)$ at time $t$.
Similar to the previous proof, $[a_J, f_J)$ is partitioned into times that an envelope subjob 
is serviced and executed, namely a \emph{busy interval}, and times that an envelope subjob is not serviced and executed, namely 
a \emph{non-busy interval}.
In consequence of the fact that both intervals are disjoint and cover the interval $[a_J, f_J)$, the response time of $J$ is given by 
the cumulative amount of time spent in either of these two states.

Now consider an envelope path $\pi_e := \setof{v_{k_1} , v_{k_2}, \dots, v_{k_p}}$ derived from $S$
we partition each arrival and finishing time interval $[a_{k_i}, f_{k_i})$ 
for each envelope subjob $i \in \setof{1, \dots, p}$ into \emph{busy} $[a_{k_i}, f_{k_i})  \cap I_{B}$ and 
\emph{non-busy} $[a_{k_i}, f_{k_i}) \cap I_{N}$ sub sets. 

\parlabel{Busy Interval.}
Clearly, the amount of cumulative \emph{busy} times is the same as in 
case of dedicated processors, namely $|[a_{k_i}, f_{k_i})  \cap I_{B}|= vol(v_{k_i})$.

\parlabel{Non-Busy Interval.}
We further partition the \emph{non-busy} interval into \emph{parallel path non-busy} if 
the envelope subjob $v_{k_i} \in V_s(\psi)$ and \emph{non-parallel path non-busy} if $v_{k_i} \in V_s^c(\psi)$. 
Since our scheduling policy is work-conserving we have that whenever an envelope subjob $v_{k_i}$ is not serviced at time $ t \in [a_{k_i}, f_{k_i})$ 
then all $m(t)$ reservations must be servicing non-envelope jobs.
We analyze the non-busy interval by cases:
  
  \begin{itemize}
    \item \textbf{Non-Parallel Path} Let by assumption $v_{k_i} \in V_s^c$ 
      then for any time $t \in [a_{k_i}, f_{k_i}) \cap I_N$ each of the $m(t)$ reservations is exclusively servicing
      subjobs from $V_s^c(\psi) \setminus v_{k_i}$ . This is due to the fact 
      that $V_s(\psi)$ subjobs have lower-priority than $V_s^c(\psi)$ subjobs and thus the servicing of 
      a subjob from $V_s(\psi)$ would imply the servicing of all pending $V_s^c(\psi)$ subjobs, which contradicts the assumption 
      that $v_{k_i} \in V_s^c$ is not serviced despite being pending.
      In consequence of this implication we have that 
      $I_{k_i}:= \setof{t \in [a_{k_i}, f_{k_i}) \cap I_N} \subseteq \setof{t \in [a_{k_i}, f_{k_i})~|~h(t) = m_i(t)}$ and thus
      \begin{equation}
        |I_{k_i}| \leq \int_{a_{k_i}}^{f_{k_i}} [h(t) = m_i(t)]~dt
       \end{equation}
       We introduce the auxiliary function $\bar{m}_i(t) = m_i-m_i(t)$ to formalize the non-service at time $t$.
       Since $(h(t)+\bar{m}_i(t))/m_i \leq 1~\forall t \in [a_{k_i},f_{k_i}) \cap I_N$, we bound
      \begin{equation}
        \label{eq:bound-non-busy-complement}
        |I_{k_i}| \leq \int_{a_{k_i}}^{f_{k_i}} \frac{h(t)+\bar{m}_i(t)}{m_i}~dt
      \end{equation}

 \item \textbf{Parallel Path}
Let by assumption the envelope subjob $v_{k_i} \in V_s$ not being serviced by any $m_i(t)$ at time $t \in [a_{k_i}, f_{k_i})$. 
Further let $[a_{i_1}, f_{i_1}), [a_{i_2}, f_{i_2}), \dots, [a_{i_z}, f_{i_z})$ denote the arrival and finishing time 
intervals of the subjobs $v_{i_1},\dots, v_{i_z}$ for each of the paths $\pi_{\psi_1}, \pi_{\psi_i},\dots, \pi_{\psi_n}$ from $\psi$ 
in the concrete schedule $S$. 
Note that, in contrast to an envelope, the arrival and finishing time intervals of those paths are not necessarily contiguous, 
i.e., $a_{i_{j+1}} \geq f_{i_j}$ for $j \in \setof{1, \dots, z-1}$ and thus requires more thought to asses when subjobs are pending. 
We use $n(t) \in \setof{1, \dots, n}$ to denote the number of pending subjobs 
from $V_s(\psi)$ at time $t$, which depends on the above arrival and finishing time intervals in $S$.
By case assumption, we know that for each $t \in [a_{k_i}, f_{k_i}) \cap I_N$ at most $n(t)-1$ 
subjobs from $V_s(\psi)$ are serviced by $m_i(t)$ at time $t$. 
Additionally, we know that pending $V_s^c(\psi)$ subjobs are prioritized before $V_s(\psi)$ subjobs, i.e.,
\begin{align}
   h(t) &\geq m_i(t)-\min\setof{m_i(t), n(t)-1} \\
  & \geq m_i(t)-(n(t)-1) \geq m_i(t)-(n-1)
\end{align}

In consequence of this implication we have that
$I_{k_i} \subseteq \setof{t \in [a_{k_i}, f_{k_i})~|~h(t) \geq m_i(t)-(n-1)}$ and thus
\begin{equation}
  |I_{k_i}| \leq \int_{a_{k_i}}^{f_{k_i}} [h(t) \geq m_i(t)-(n-1)]~dt
\end{equation}
Reusing the auxiliary function $\bar{m}_i(t) = m_i-m_i(t)$ and the fact 
that 
\begin{equation}
 \frac{h(t)+\bar{m}_i(t)}{m_i-n(t)+1} \geq 1 ~\forall t \in [a_{k_i}, f_{k_i}) \cap I_N
\end{equation} we conclude that 
\begin{equation}
\label{eq:bound-non-busy-collection}
  |I_{k_i}| \leq \int_{a_{k_i}}^{f_{k_i}} \frac{h(t)+\bar{m}(t)}{m_i-(n-1)}~dt
\end{equation}
\end{itemize}

In conclusion we have that
\begin{equation}
\label{eq:non-busy-combined}
 |I_N| \leq \sum_{i=1}^p |I_{k_i}| := 
\begin{cases} 
  \text{Eq.~\eqref{eq:bound-non-busy-collection}} & \text{if $v_{k_i} \in V_s(\psi)$} \\
  \text{Eq.~\eqref{eq:bound-non-busy-complement}} & \text{if $v_{k_i} \in V_s^c(\psi)$}
\end{cases}
\end{equation} and since Eq.~\eqref{eq:bound-non-busy-collection} $\geq$ Eq.~\eqref{eq:bound-non-busy-complement}, 
we conclude that 

\begin{equation}
\label{eq:non-busy-combined-bound}
 |I_N| \leq \sum_{i=1}^p \int_{a_{k_i}}^{f_{k_i}} \frac{h(t)+\bar{m}(t)}{m_i-(n-1)}~dt = \int_{a_J}^{f_J} \frac{h(t)+\bar{m}(t)}{m_i-(n-1)}~dt
\end{equation}

By definition of $\int_{a_J}^{f_J} h(t)~dt = vol(V_s^c(\psi))$ we conclude that
\begin{equation}
|I_N| \leq \frac{vol(V_s^c(\psi))}{m_i-n+1} + \int_{a_J}^{f_J} \frac{\bar{m}_i(t)}{m_i-n+1}~dt
\end{equation}

In conclusion, we know that $R_j = |I_B| + |I_N|$
\begin{equation}
  \label{eq:final-proof-a}
  R_J \leq vol(\pi_*)  + \frac{vol(V_s^c(\psi))+\int_{a_J}^{f_J} \bar{m}_i(t)~dt}{m_i-n+1}
\end{equation}
The contract of the reservation system for a job of DAG task $\tau_i$ promises 
to provide $E_i^1, \dots, E_i^{m_i}$ service during the arrival time of the DAG job $J$ and its deadline $[a_J, a_J+D_{i}) \supseteq [a_J, f_J)$.
Therefore, each of the $m_i$ individual reservations does not provide service for at most $D_{i}-E_{i}^p$ for $p \in \setof{1, \dots, m_i}$ 
amount of time, which implies that
\begin{equation}
  \label{eq:final-proof-b}
 \int_{a_J}^{f_J} \bar{m}_i(t)~dt \leq \sum_{p=1}^{m_i} D_{i}-E_{i}^p
\end{equation}
Furthermore, injecting Eq.~(\ref{eq:final-proof-b}) into Eq.~(\ref{eq:final-proof-a}),
\begin{align*}
  R_J \leq &  vol(\pi_*)  + \frac{vol(V_s^c(\psi))+m_i D_i -\sum_{p=1}^{m_i} E_{i}^p}{m_i-n+1}  \\
  = & D_i+vol(\pi_*)  + \frac{vol(V_s^c(\psi))+(n-1)D_i -\sum_{p=1}^{m_i} E_{i}^p}{m_i-n+1} \\
  \leq & D_i \qquad\qquad(\mbox{due to Eq.~(\ref{eq:res-theorem})})
\end{align*}
\end{proof}

\begin{figure}[tb]
\centering
\resizebox{.9\linewidth}{!}{\begin{tikzpicture}[]
\def\ux{1.7cm}\def\uy{1.20cm}
    \def\dx{0.88cm}
        \tikzset{
          job/.style={fill=white, text=black, very thick, minimum height=8mm, draw},
          interference/.style={pattern=north east lines, text=black, minimum height=7mm},
          collection/.style={ text=black, minimum height=8mm, very thick},
          nonservice/.style={fill=black!8, very thick, draw, text=black, minimum height=8.5mm},
          service/.style={fill=black!20, text=black, minimum height=8.5mm},
	  envelope/.style={fill=gray!40,  text=black, very thick, minimum height=8mm, draw },
          prec/.style={color=black, very thick}
        }

        \begin{scope}[shift={(0, 7.5)}]
          
           \node[service, minimum width=2.5*\dx, anchor=south west] at (0*\dx,0) (service) {};
           \node[service, minimum width=2*\dx, anchor=south west] at (4*\dx,0) (service) {};
                   \node[service, minimum width=9*\dx, anchor=south west] at (6*\dx,0) (service) {};
          
           \foreach \x in {0, 1,...,16} 
           \draw[-, very thick](\x*\dx, 0.1*\dx) -- (\x*\dx, -0.1*\dx)node[below] {\large $\x$};
           \draw[-, very thick] (0,0) node[anchor=south east] {}-- coordinate (xaxis) (16*\dx,0);
           \draw[->, very thick](0*\dx, 0.1*\dx) -- (0*\dx, 1.2*\dx)node[below] {};
           \draw[<-, very thick](16*\dx, 0.1*\dx) -- (16*\dx, 1.2*\dx)node[below] {};

           \node[collection, minimum width=1*\dx, anchor=south west,draw] at (6*\dx,0) (v1) {\LARGE $v_3$};
            \node[collection, minimum width=1*\dx, anchor=south west,draw] at (8*\dx,0) (v2) {\LARGE $v_9$};
        \end{scope}

	\begin{scope}[shift={(0, 5.5)}] 
          
           \node[service, minimum width=3*\dx, anchor=south west] at (0*\dx,0) (service) {};
           \node[service, minimum width=10.5*\dx, anchor=south west] at (4*\dx,0) (service) {};
           
           \foreach \x in {0, 1,...,16} 
           \draw[-, very thick](\x*\dx, 0.1*\dx) -- (\x*\dx, -0.1*\dx)node[below] {\large $\x$};
           \draw[-, very thick] (0,0) node[anchor=south east] {}-- coordinate (xaxis) (16*\dx,0);
           \draw[->, very thick](0*\dx, 0.1*\dx) -- (0*\dx, 1.2*\dx)node[below] {};
           \draw[<-, very thick](16*\dx, 0.1*\dx) -- (16*\dx, 1.2*\dx)node[below] {};
           
           \node[collection, minimum width=3*\dx, anchor=south west,draw] at (0*\dx,0) (v1) {\LARGE $v_1$};
           \node[collection, minimum width=2*\dx, anchor=south west,draw] at (4*\dx,0) (v1) {\LARGE $v_7$};
           \node[collection, minimum width=2*\dx, anchor=south west,draw] at (6*\dx,0) (v1) {\LARGE $v_5$};
           \node[collection, minimum width=3*\dx, anchor=south west,draw] at (8*\dx,0) (v1) {\LARGE $v_6$};
        \end{scope}

        \begin{scope}[shift={(0, 3.5)}] 
           \node[service, minimum width=9*\dx, anchor=south west] at (1*\dx,0) (service) {};
           \node[service, minimum width=4.5*\dx, anchor=south west] at (10.5*\dx,0) (service) {};
          
           \foreach \x in {0, 1,...,16} 
           \draw[-, very thick](\x*\dx, 0.1*\dx) -- (\x*\dx, -0.1*\dx)node[below] {\large $\x$};
           \draw[-, very thick] (0,0) node[anchor=south east] {}-- coordinate (xaxis) (16*\dx,0);
           \draw[->, very thick](0*\dx, 0.1*\dx) -- (0*\dx, 1.2*\dx)node[below] {};
           \draw[<-, very thick](16*\dx, 0.1*\dx) -- (16*\dx, 1.2*\dx)node[below] {};

           \node[collection, minimum width=1*\dx, anchor=south west,draw] at (3*\dx,0) (v1) {\LARGE $v_4$};
           \node[collection, minimum width=2*\dx, anchor=south west,draw] at (6*\dx,0) (v1) {\LARGE $v_8$};
           
        \end{scope}

        \begin{scope}[shift={(0, 1.5)}] 
          
           \node[service, minimum width=4.5*\dx, anchor=south west] at (1.5*\dx,0) (service) {};
           \node[service, minimum width=9*\dx, anchor=south west] at (7*\dx,0) (service) {};
          
           \foreach \x in {0, 1,...,16} 
           \draw[-, very thick](\x*\dx, 0.1*\dx) -- (\x*\dx, -0.1*\dx)node[below] {\large $\x$};
           \draw[-, very thick] (0,0) node[anchor=south east] {}-- coordinate (xaxis) (16*\dx,0);
           \draw[->, very thick](0*\dx, 0.1*\dx) -- (0*\dx, 1.2*\dx)node[below] {};
           \draw[<-, very thick](16*\dx, 0.1*\dx) -- (16*\dx, 1.2*\dx)node[below] {};

           \node[collection, minimum width=3*\dx, anchor=south west,draw] at (3*\dx,0) (v1) {\LARGE $v_2$};
           
        \end{scope}

           \begin{scope}[shift={(0, 0)}] 
                \node[service, minimum height=8mm, minimum width=8mm, rectangle, label=0:\LARGE Service] at (4, 0) {};
                \node[job, minimum height=8mm, minimum width=8mm, rectangle, label=0:\LARGE No Service] at (8,0) {};
              \end{scope}

    \end{tikzpicture}

}
\caption{Exemplary schedule for an ordinary reservation system of the DAG illustrated in Figure~\ref{fig:example-dag-plain} 
  with a deadline of $16$. A reservation system that consists of $4$ equally sized reservations of $13.5$ time units 
  using a $3$- greedy path collection.}
\label{fig:example-schedule-ordinary-reservation}
\end{figure}
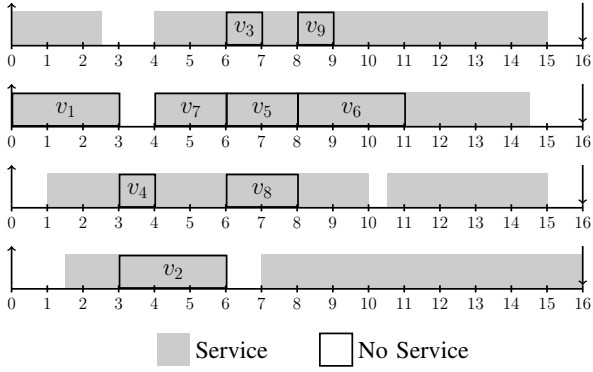

\begin{algorithm}[tb]
\caption{Minimal Service $m_i$-Ordinary Reservations}
\label{alg:ordinary-reservation-design}
\begin{algorithmic}[1]
\Require {DAG Task $\tau_i := (G_{i}, D_{i}, T_{i})$, No. CPU $M$, $vol_i$.}
\Ensure Minimal Service $m_i$-ordinary Reservations $\mathcal{O}_i$.
\State {$\xi \leftarrow [0, \dots, 0]$;}
\For{\textbf{each} $n \in \setof{1, \dots, M}$}
  \State {search the longest path $\pi_*$ in $G_i$ with respect to $vol$;}
  \If{$vol_i(\pi_*) > 0$} \Comment{Some subjobs not yet visited}
     \State{$\xi[n] \leftarrow \xi[n-1] + vol_i(\pi_*)$;}
     \State {set $vol_i(v) := 0$ for each $v$ in $\pi_*$;}
   \Else
   \State{break;}
\EndIf
\EndFor
\State{$E_{\sigma} \leftarrow (\emptyset, \infty);$} \Comment{$((m_i,n), \sum_{p=1}^{m_i} E_{i}^p)$}
\For{\textbf{each} $m_i \in \setof{1, \dots, M}$}
 \For{\textbf{each} $n \in \setof{1, \dots, m_i}$}
 \State{$E \leftarrow (m_i-n+1) \cdot \xi[1] + (n-1) \cdot D_{i} + C_i-\xi[n]$;}
 \If{$E \geq m_i \cdot D_i$ or $E \leq m_i \cdot \xi[1]$}
 \State{continue;}
 \EndIf
  \If{$E < E_{\sigma}[2]$;}
    \State{$E_{\sigma}[1] \leftarrow (m_i, n)$ and $E_{\sigma}[2] \leftarrow E$;}
  \EndIf
 \EndFor
\EndFor
\State{\textbf{return} $\mathcal{O}_i := \setof{(E_{\sigma}[1], D_i, T_i), \dots, (E_{\sigma}[1], D_i, T_i)}$;}
\end{algorithmic}
\end{algorithm}

An exemplary $4$-ordinary reservation system using a \mbox{$3$-} 
path collection for the DAG shown in Figure~\ref{fig:example-dag-plain} with deadline $16$
is illustrated in Figure~\ref{fig:example-schedule-ordinary-reservation}.
In this example, each reservation is equal in size which results to  $E_i^p = 13.5$ for $p \in \setof{1, \dots, 4}$ 
according to Eq.~\eqref{eq:res-theorem} with $D_i = 16$, $n=3$, $m_i = 4$, and $vol(V_s^c(\psi)) = 4$.
Please notice that the service can be provided arbitrarily depending on a concrete schedule 
as long as the promised service is provided in the arrival and deadline interval.
Similar to the problem of gang reservation provisioning, we search for feasible provisions 
that minimize the total service $\sum_{p=1}^{m_i} E_{i}^p$  under the constraint that 
$vol(\pi_*) \leq E_i^p \leq D_i$ for $p \in \setof{1, \dots, m_i}$ by exhaustive search and 
the greedy heuristic for the $n$- path collection problem as shown in Algorithm~\ref{alg:ordinary-reservation-design}.

\section{Evaluation}
\label{sec:evaluation}

The objectives of our evaluations are twofold, namely to first asses the performance strengths and weaknesses 
of the parallel path progression concepts compared to state-of-the-art single path approaches as represented by 
He et al.~\cite{DBLP:conf/ecrts/HeLG21} with respect to the makespan problem.
Secondly, we evaluate our proposed gang and ordinary reservation systems provisioning 
strategies from Algorithm~\ref{alg:gang-reservation-design} and Algorithm~\ref{alg:ordinary-reservation-design} 
for resource over-provisioning with respect to a lower-bound. 

In our hierarchical scheduling approach, the schedulability depends on the 
schedulability analysis used for the reservations' task models and the over-provisioning ratio.
Since we are only interested in the evaluation of the reservations systems themselves, 
we abstain from schedulability acceptance ratio experiments and focus on the over-provision ratio.
Throughout this section, we use \emph{HE} to refer to He et al.~\cite{DBLP:conf/ecrts/HeLG21}, 
\emph{FED} to refer to federated scheduling from Li et al.\cite{Li:ECRTS14}, \emph{OUR}, \emph{OUR-GANG}, 
and \emph{OUR-ORD} to refer to the bounds in Eq.~\eqref{eq:fed-response}, Eq.~\eqref{eq:gang-res}, 
and Eq.~\eqref{eq:res-theorem}, respectively.

\parlabel{Directed-Acyclic Graph Generation.}
Motivated by the fact that the internal structure of the DAG under evaluation 
strongly impacts the performance of the evaluated analyses, a parameterized generation 
process is used to randomly generate DAGs whose structure that can be attributed to the 
specified parameters. To that end,  we use a layer-by-layer DAG generation process 
with the parameters \emph{parallelism}, \emph{min layer}, \emph{max layers}, 
and \emph{connection probability}.
In a generation step, the number of layers is chosen uniform at random from the range 
\emph{min layer} to \emph{max layer}. In each layer, the number of subtasks 
is drawn uniform at random from the range $1$ to \emph{parallelism}.
The connection of subtasks at a layer $\ell$ is only allowed by subtasks 
from layer $\ell-1$. 
Each newly generated subtask is connected with a subtask from the previous layer 
with probability \emph{connection probability}.
Each subtask is assigned an integer worst-case execution time drawn uniform at random 
from the range $1$ to $100$.

\subsection{Makespan Experiment}
\label{evaluation:makespan-experiment}

We generated 100 DAGs with $5$-$10$ layers for each experiment 
with the following varying parameters; \emph{parallelism} in $\setof{2, 4, 8, 16, 20}$ and 
\emph{connection probability} in $\setof{.2, .6, .8}$, 
and evaluate the makespan, i.e., the worst-case response time of a single DAG job 
on $M$ processors exclusively, where $M$ is in $\setof{2, 4, 8, 16}$.
A representative selection of the results is shown in the box plots in Figure~\ref{fig:boxplot-8-20-2}, Figure~\ref{fig:boxplot-4-80-8}, 
Figure~\ref{fig:boxplot-8-80-16}, and Figure~\ref{fig:boxplot-20-20-16}.  
where the makespan of \emph{HE}, \emph{FED}, and \emph{OUR} is normalized to 
a theoretical lower-bound of $\max\setof{C/M, vol(\pi_*)}$, i.e., $100 \%$ implies a tight result.

\begin{figure}[tb]
\centering
\includegraphics[width=0.43\textwidth]{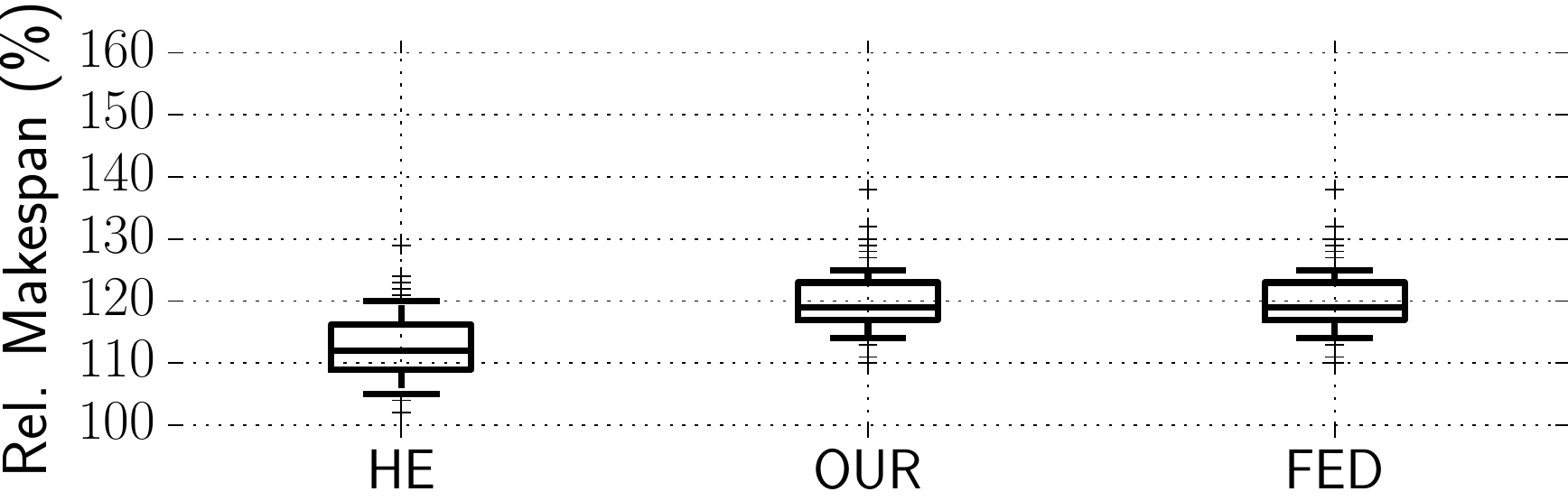}
\caption{Rel.makespan for $100$ DAGs generated with \emph{parallelism} $8$, \emph{connection probability} $20\%$ 
on $2$ processors.}
\label{fig:boxplot-8-20-2}
\end{figure}

\begin{figure}[tb]
\centering
\includegraphics[width=0.43\textwidth]{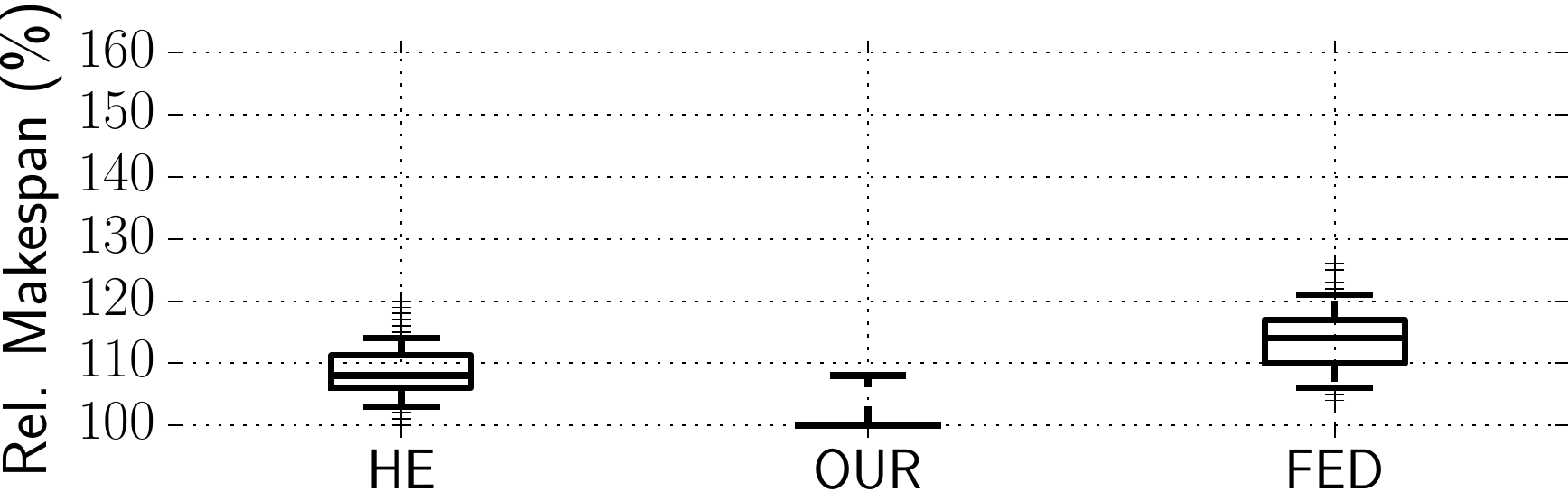}
\caption{Rel.makespan for $100$ DAGs generated with \emph{parallelism} $4$, \emph{connection probability} $80\%$ 
on $8$ processors.}
\label{fig:boxplot-4-80-8}
\end{figure}

\begin{figure}[tb]
\centering
\includegraphics[width=0.43\textwidth]{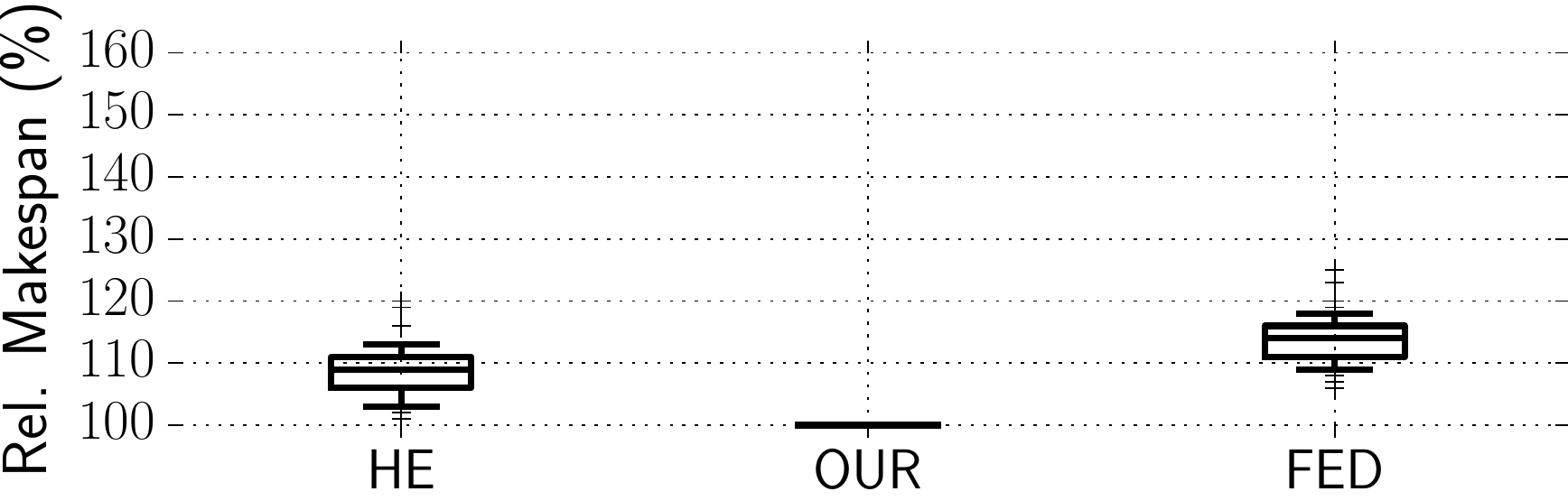}
\caption{Rel.makespan for $100$ DAGs generated with \emph{parallelism} $8$, \emph{connection probability} $80\%$ 
on $16$ processors.}
\label{fig:boxplot-8-80-16}
\end{figure}

\begin{figure}[tb]
\centering
\includegraphics[width=0.43\textwidth]{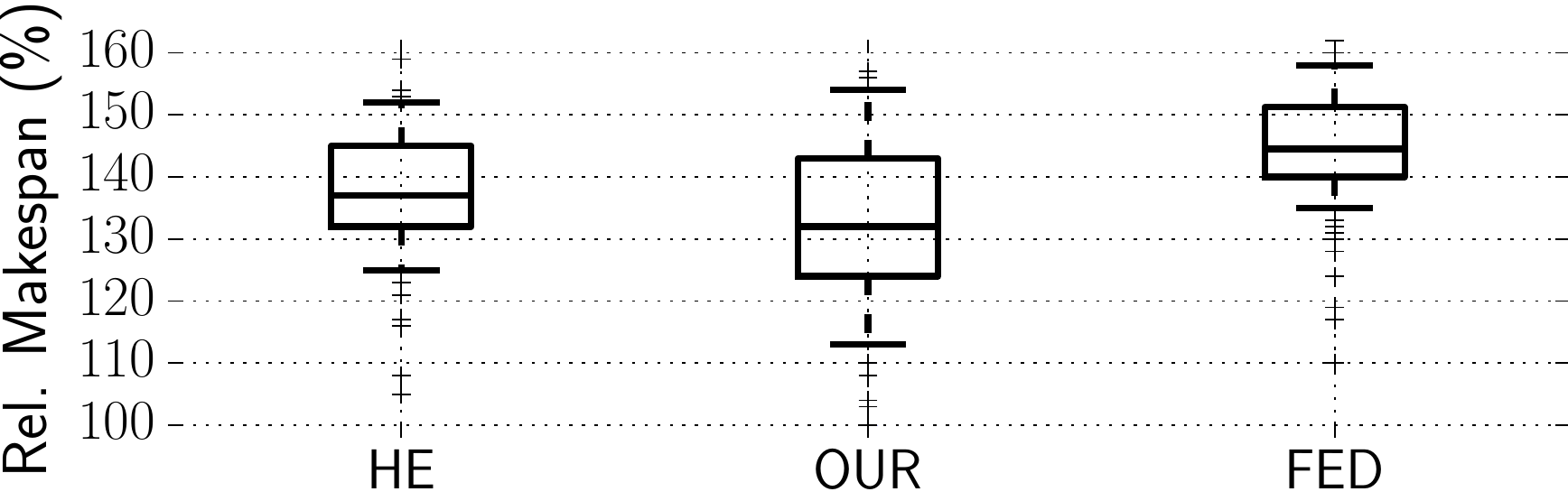}
\caption{Rel.makespan for $100$ DAGs generated with \emph{parallelism} $20$, \emph{connection probability} $20\%$ 
on $16$ processors.}
\label{fig:boxplot-20-20-16}
\end{figure}

\parlabel{Observation (Low-Parallelism).}
As shown in Figure~\ref{fig:boxplot-8-20-2} for the DAG set that was 
generated with the parameters \emph{parallelism} $8$ and \emph{connection probability} $20\%$ 
on $2$ processors, the analysis \emph{HE} yields an average relative makespan of roughly $112 \%$, 
whilst \emph{OUR} approach degenerates to federated scheduling \emph{FED} with roughly $119 \%$.
Due to the limited number of processors, our approach can not benefit from the inherent parallelism 
of the DAGs. 

The other case, in which the parallel execution is limited by the number of processors, 
is shown in Figure~\ref{fig:boxplot-20-20-16}. Here, the \emph{parallelism} of $20$ 
is larger than the number of processors $16$, and
\emph{OUR} 
approach outperforms \emph{HE} on average by as much as $5\%$. 
This case demonstrates that \emph{OUR} approach can leverage the 
parallelism offered by the processors and the inherent parallelism of the DAGs 
to improve the response time.

\parlabel{Observation (High-Parallelism).}
The largest performance improvements of \emph{OUR} can be observed 
for highly parallel DAGs in combination with high processor availability. This case is  
shown in Figure~\ref{fig:boxplot-8-80-16} with \emph{parallelism} of $8$ 
on $16$ processors and Figure~\ref{fig:boxplot-4-80-8} with 
\emph{parallelism} of $4$ on $8$ processors. 
We observe that \emph{OUR} approach not only outperforms 
\emph{HE} in both cases, but that \emph{OUR} approach even yields tight results for most evaluated DAGs.

\subsection{Reservations Over-Provisioning Experiment}
\label{sec:reservations-over-provisioning-experiment}

\begin{table}[htb]
\setlength{\arrayrulewidth}{.15em}
\setlength\extrarowheight{2.2pt}
\centering
\begin{tabular}{l|c|c|c}
  \textsc{Waste Ratio.} (\%)& \textsc{P=4} & \textsc{P=8} & \textsc{P=16}\\ \hline
 \textsc{OUR-GANG min.} & 27.26 & 29.22 & 36.22 \\
 \textsc{OUR-GANG max.} & 76.85 & 82.33 & 82.19 \\
 \textsc{OUR-GANG average} & 53.76 & 59.26 & 61.97 \\
 \textsc{OUR-GANG variance} & 87.06 & 74.42 & 62.14 \\ \hline
 \textsc{OUR-ORD min.} & 13.79 & 21.89 & 36.87 \\
 \textsc{OUR-ORD max.} & 76.64 & 81.77 & 82.07 \\
 \textsc{OUR-ORD average} & 53.33 & 59.54 & 62.4 \\
 \textsc{OUR-ORD variance} & 97.17 & 73.72 & 58.01 \\ \hline
\end{tabular}
\vspace{0.8em}
\caption{Minimal over-provisioning ratio for DAG sets with different degrees of parallelism $P$ 
on a system with infinite processors.}
\label{table:over-allocation-ratio}
\end{table}

We compare \emph{OUR-GANG} in Algorithm~\ref{alg:gang-reservation-design} and \emph{OUR-ORD} in Algorithm~\ref{alg:ordinary-reservation-design} 
in terms of an waste ratio, i.e., the reserved service minus the total workload divided by the reserved service. 

We generated 100 DAGs with $5$-$10$ for each of the following varying parameters: \emph{parallelism} in $\setof{4, 8, 16}$, 
\emph{connection probability} in $\setof{.1, .2, .3}$,   
deadline such that $vol(\pi_*) < D < \min\setof{\rho \cdot vol(\pi_*), C}$ 
for $\rho \in \setof{1.2, 1.4, 1.6, 1.8}$.  
We assume that a sufficient number of processors 
is available such that a reservation system can be found for every generated deadline.
We then evaluated all generates sets with the same parallelism 
$P$ together and present the results in Table.~\ref{table:over-allocation-ratio}. 

\parlabel{Observation.}
The overall statistics of both reservation systems are quite similar for each 
level of parallelism. Especially, the average values only differ 
by $0.42 \%$. For both, the waste ratio slightly increases with the amount 
of parallelism,  
resulting from the increased number of paths in a path collection 
(and thus at least a similar increase in the number of reservations) 
to exploit the parallelism and satisfy deadline constraints.

\subsection{Reservations Improvement Experiment}

Lastly, we compare our ordinary reservation system to the state-of-the-art hierarchical DAG scheduling approach proposed 
by Ueter et al.~\cite{Ueter2018}. Recall that \emph{OUR-ORD} is a generalization of their approach. That is,  
\emph{OUR-ORD}, as generated by Algorithm~\ref{alg:ordinary-reservation-design}, 
always has a higher resource efficiency,
since a solution of Ueter et al.~\cite{Ueter2018} is identical to our 
approach if $n=1$.
However, to what extend resources in terms of the number of reservations and the overall service 
can be saved is not clear beforehand and thus subject to our second experiment.
We use the same setup as  
described in Section~\ref{sec:reservations-over-provisioning-experiment} 
and evaluate the reservation system as generated by Ueter et al.~\cite{Ueter2018} referred to as \emph{UET} 
against \emph{OUR-ORD} with respect to the overall reserved service and the number of parallel reservations 
for DAGs with parallelism $P \in \setof{4, 8, 16}$.

\begin{figure}[tb]
\centering
\includegraphics[width=0.43\textwidth]{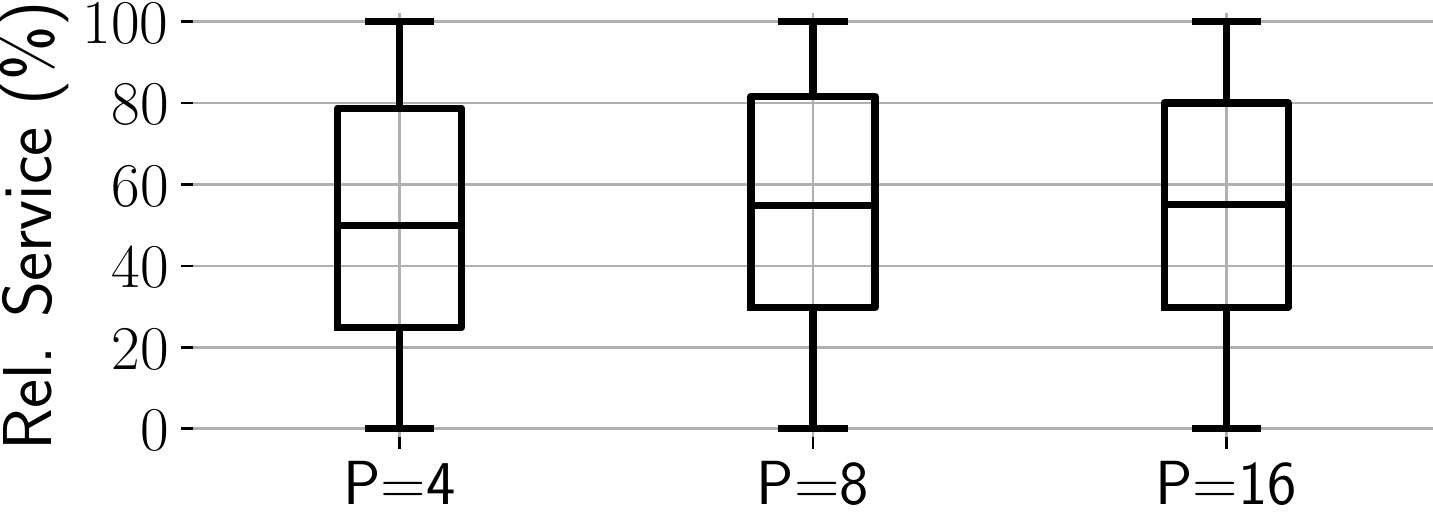}
\caption{Reserved service of \emph{OUR-ORD} normalized to \emph{UET} for $100$ DAGs generated with \emph{parallelism} $P \in \setof{4, 8, 16}$.}
\label{fig:service-comparison}
\end{figure}

\begin{figure}[tb]
\centering
\includegraphics[width=0.43\textwidth]{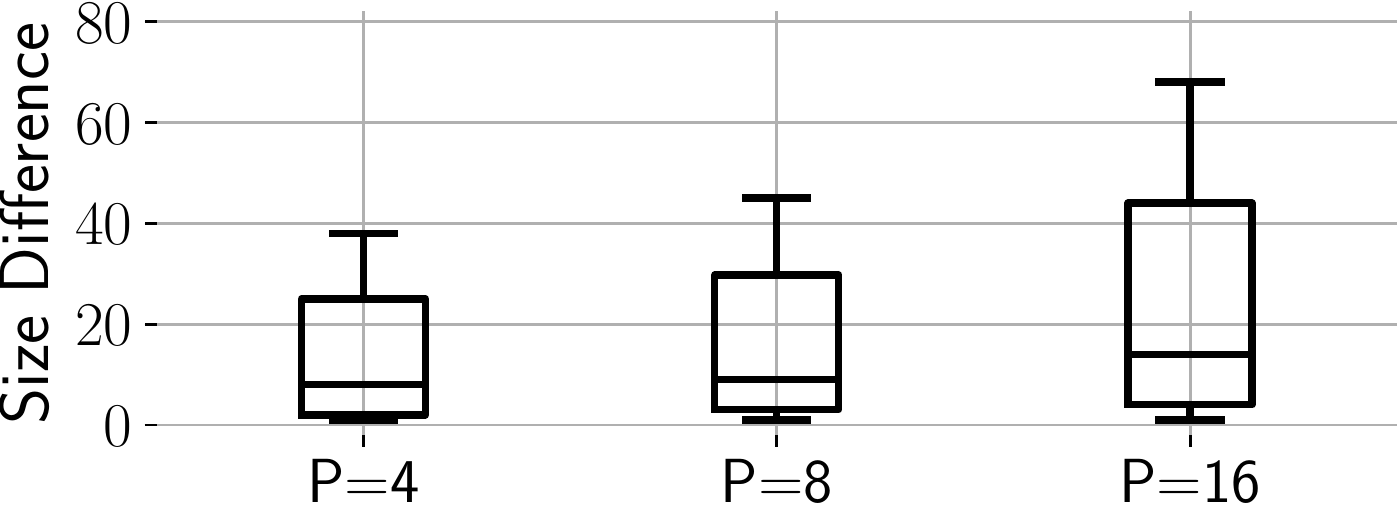}
\caption{Absolute difference of the number of reservations of \emph{UET} minus \emph{OUR-ORD} for \emph{parallelism} $P \in \setof{4, 8, 16}$.}
\label{fig:res-size-comparison}
\end{figure}

\begin{table}[tb]
\setlength{\arrayrulewidth}{.15em}
\setlength\extrarowheight{2.2pt}
\centering
\begin{tabular}{l|c|c|c}
  \textsc{Size Difference}& \textsc{P=4} & \textsc{P=8} & \textsc{P=16}\\ \hline
 \textsc{min.} & 1 & 1 & 1 \\
 \textsc{max.} & 839 & 2035 & 3489 \\
 \textsc{average} & 30.42 & 48.57 & 64.72 \\
 \textsc{variance} & 5404.84 & 21760.85 & 43247.11 \\ \hline
\end{tabular}
\vspace{0.8em}
\caption{Statistics of the absolute reservation size differences for different degrees of parallelism $P$.}
\label{table:res-size-differences}
\end{table}

\parlabel{Observation (Service).}
Figure~\ref{fig:service-comparison} shows 
that the reserved service of \emph{OUR-ORD} saves roughly $50\%$ 
on average for each level of parallelism, while the improvements 
are larger for smaller parallelism.

\parlabel{Observation (Size).}
In Figure~\ref{fig:res-size-comparison} and Table~\ref{table:res-size-differences}, 
the results for the absolute differences in the reservation sizes, i.e., \emph{UET} minus \emph{OUR-ORD}, 
is shown.
Table~\ref{table:res-size-differences} shows that \emph{OUR-ORD} can significantly decrease the number of required 
parallel reservations to meet deadline constraints by roughly $30,~49,$ and $65$ reservations.
Interestingly, the average difference is relatively stable compared to the range and
variance for increasing degree of parallelism.
This demonstrates that our approach can improve \emph{UET} to be applicable
for highly parallel DAG tasks with tight deadline constraints.

\section{Related Work}
\label{sec:related-work}

Real-time aware scheduling of parallel task systems has been 
extensively studied for a variety of different proposed task models. 
Goosens et al.~\cite{DBLP:journals/corr/abs-1006-2617} have provided a 
classification of parallel task with real-time constraints.

Early work on parallel task models focused on synchronous parallel task models,
e.g.,~\cite{DBLP:conf/rtns/MaiaBNP14, SaifullahRTSS2011, DBLP:conf/ecrts/ChwaLPES13}. 
Synchronous parallel task models extend the fork-join model~\cite{Conway63} in such a way that 
they allow different numbers of subtasks in each (synchronized) segment where 
the number of subtasks can exceed the number of processors. 

A prominent parallel task model that has been subject to many recent scheduling and analysis 
efforts, is the directed-acyclic graph (DAG) task model.
The DAG models subtask-level parallelism by acyclic precedence constraints 
for a set of subtasks. This parallel model has been shown to correspond to models available 
in parallel computing APIs such as OpenMP by Melani et al.~\cite{DBLP:conf/cases/SerranoMVMBQ15}, Sun et al.~\cite{DBLP:conf/rtss/SunGWHY17}, 
or Serrano et al.~\cite{Serrano2018}. The parallel DAG task model has been studied for global~\cite{DBLP:conf/ecrts/BonifaciMSW13,nasri_et_al:LIPIcs:2019:10758, Chen+Agrawal2014} and 
partitioned scheduling~\cite{DBLP:conf/sies/FonsecaNNP16, 8603232, Amor2019, Fonseca:2017:RTNS}.

The proposed scheduling algorithms and analyses in the literature can be 
categorized into decompositional and non-decompositional. 
In the former, the parallel task model is decomposed into a set of sequential task models, which 
is scheduled and analyzed in their stead, e.g.,~\cite{JiangTCAD}. Non-decompositional approaches consider the peculiarities 
of the parallel task models, e.g.,~\cite{Li:ECRTS14, Ueter2018, DBLP:conf/date/Baruah15, nasri_et_al:LIPIcs:2019:10758, DongRTSS, DBLP:conf/ecrts/BonifaciMSW13}.

A prominent decomposition based approach is federated scheduling by Li et al.~\cite{Li:ECRTS14}
that avoids inter-task interference for parallel tasks. It has been extended in, e.g., \cite{RTSS-2017-semi-federated, Jiang2021, DinhRTCSA2020, Ueter2018, DBLP:conf/date/Baruah15, DBLP:conf/ipps/Baruah15}.
In the original federated scheduling approach, the set of DAG tasks is partitioned into
tasks that can be executed sequentially on a single processor and tasks that need
to execute in parallel to finish before their respective deadlines.

\section{Conclusion and Future Work}
\label{sec:conclusion-and-future-work}

We present the parallel path progression concept that allows 
to improve the self-interference analysis by explicitly 
considering the parallel execution of paths in the DAG. 
We propose a sustainable scheduling algorithm and analysis that 
is extended by virtue of 
hierarchical scheduling for gang-based and ordinary reservation systems for sporadic 
arbitrary-deadline DAG tasks. We designed an approximation algorithm for the $n$-path 
collection and proved that the then resulting makespan in our parallel path progression 
scheduling algorithm yields upper-bounds with respect to an optimal solution.
For these reservations, we provide heuristic algorithms that provision 
reservation systems with respect to the service they require.
We evaluated our approach using synthetically generated DAG task sets 
and demonstrated that our approach can improve the state-of-the art in high-parallelism 
scenarios while demonstrating reasonable performance for low-parallelism scenarios.
In future work, we plan to improve the active idling issues of the proposed reservation 
systems by considering self-suspension behaviour for the reservation systems. 

\section*{Acknowledgement}
This result is part of a project~(PropRT) that has received funding
from the European Research Council~(ERC) under the European Union’s
Horizon 2020 research and innovation programme~(grant agreement
No. 865170).  This work has been supported by Deutsche
Forschungsgemeinschaft~(DFG), as part of Sus-Aware~(Project
No. 398602212).
        
\bibliography{real-time}
\bibliographystyle{abbrv}

\end{document}